\pdfoutput=1
\documentclass{article} %
\usepackage{iclr2027_conference,times}

\usepackage{amsmath,amssymb,amsthm,mathtools}
\usepackage{booktabs}
\usepackage{graphicx}
\usepackage{algorithm}
\usepackage{algpseudocode}
\usepackage{enumitem}
\usepackage{xcolor}
\usepackage{hyperref}
\usepackage{url}

\theoremstyle{plain}
\newtheorem{theorem}{Theorem}[section]
\newtheorem{lemma}[theorem]{Lemma}
\newtheorem{proposition}[theorem]{Proposition}
\newtheorem{corollary}[theorem]{Corollary}
\theoremstyle{definition}

\newtheorem{assumption}[theorem]{Assumption}
\newtheorem{example}[theorem]{Example}
\theoremstyle{remark}
\newtheorem{remark}[theorem]{Remark}

\newcommand{\E}{\mathbb{E}}
\renewcommand{\Pr}{\mathbb{P}}
\newcommand{\ind}[1]{\mathbf{1}\{#1\}}
\newcommand{\F}{\mathcal{F}}          %
\newcommand{\N}{\mathbb{N}}
\newcommand{\R}{\mathbb{R}}

\newcommand{\M}{M}                    %
\newcommand{\Mset}{[\M]}
\newcommand{\thetastar}{\theta^\star}
\newcommand{\rank}{R}                 %
\newcommand{\pairs}{\mathcal{P}}      %

\newcommand{\Hyp}{H}                  %
\newcommand{\W}{\mathcal{W}}          %
\newcommand{\wo}{W}                   %
\newcommand{\wostar}{W^\star}         %
\newcommand{\T}{T}                    %
\newcommand{\wle}{\preceq}            %
\newcommand{\wlt}{\prec}
\newcommand{\weq}{\sim}

\newcommand{\Ep}{E}                   %
\newcommand{\CS}{\mathcal{C}}         %
\newcommand{\D}{\mathcal{D}}          %
\newcommand{\Dbar}{\overline{\mathcal{D}}} %
\newcommand{\good}{\mathcal{G}}       %
\newcommand{\Rset}{\mathcal{R}}       %
\newcommand{\tier}{\operatorname{tier}}
\newcommand{\hgt}{h}

\allowdisplaybreaks

\title{Rank Confidence Sequences: \\ Anytime-Valid Leaderboards}

\author{Hamed Khosravi and Xiaoming Huo}
\makeatletter
\def\@maketitle{\vbox{\hsize\textwidth\centering
{\LARGE\sc \@title\par}
\vskip 0.3in
\begin{tabular}[t]{c}\textbf{Hamed Khosravi}\\ \texttt{hkhosravi7@gatech.edu}\end{tabular}\hspace{4em}%
\begin{tabular}[t]{c}\textbf{Xiaoming Huo}\\ \texttt{huo@gatech.edu}\end{tabular}\par
\vskip 0.12in
H.~Milton Stewart School of Industrial and Systems Engineering\\
Georgia Institute of Technology, Atlanta, GA 30332, USA\par
\vskip 0.3in minus 0.1in}}
\makeatother
\hypersetup{pdftitle={Rank Confidence Sequences: Anytime-Valid Leaderboards},pdfauthor={Hamed Khosravi, Xiaoming Huo}}

\iclrfinalcopy %
\begin{document}

\maketitle
\lhead{}\renewcommand{\headrulewidth}{0pt} %
\addtocontents{toc}{\protect\setcounter{tocdepth}{-1}}

\begin{abstract}
Leaderboards rank models by their average scores on benchmark items, and they are consulted
repeatedly while the evaluation is still running. Existing confidence intervals for a model's rank
control their error rate only if they are computed once, after a number of items chosen in
advance. If they are recomputed as results arrive, and the evaluation stops once they look
decisive, their error rate exceeds its nominal level. Anytime-valid methods keep their guarantees
at all sample sizes simultaneously and hence under any stopping rule. They exist for the accuracy of
one model, for one pair of models and for the set of models that may be best. For pairwise battles
they also give ranks. None gives ranks
when all models are scored on the same items, which makes their scores dependent. We construct
rank confidence sequences: for every model, a set of ranks that contains its true rank,
simultaneously for all models and at all times, at a chosen error level $\alpha$, in finite
samples. The construction combines betting e-processes, one for each ordered pair of models, with
closed testing over the possible orderings of the models. It allows any dependence between
the models' scores on an item. The method has two advantages. A leaderboard can be inspected after every item
without inflating its error rate. The evaluation of each model can stop as soon as the question
asked about it is answered, which saves compute. When results are examined only once, halfway through or later, little power
is lost relative to fixed-sample methods. The paper quantifies these advantages in simulations and
on public leaderboard data.

\end{abstract}

\section{Introduction}
\label{sec:intro}

\paragraph{The problem.}
Confidence intervals for leaderboard ranks are designed to be computed once, but leaderboards are
consulted many times while the evaluation is still running. A benchmark leaderboard scores each
of $\M$ models on a common set of test items, such as exam questions, and ranks the models by
their average scores. The averages come from finitely many items, so the ranking is uncertain.
Several methods therefore attach a confidence interval to each model's rank
\citep{mogstad2024,almohamad2022,neuhof2024confident}. At confidence level $95\%$, these intervals contain the true
ranks of all models simultaneously with probability at least $95\%$, some only asymptotically. The guarantee assumes a
specific protocol. The number of items is chosen before any results are seen. The intervals are
then computed once, after all of those items have been scored. We call such intervals
\emph{fixed-sample}.

\paragraph{Repeated looks.}
Inspecting a leaderboard repeatedly as results arrive breaks the fixed-sample guarantee. In
practice, models are scored item by item or batch by batch, and the standings are examined along
the way. We call each examination a \emph{look}. Because evaluation is expensive \citep{evaleval2026}, it is cut short when the ranking seems
clear and extended only when two models seem close. Each look, however, is another chance
for a fixed-sample interval to be wrong. Stopping when the intervals seem decisive also tends
to stop exactly when one of them is wrong. The probability that at least one look produces a false
statement therefore exceeds the nominal level. This effect of repeated significance tests has long
been known, and \citet{armitage1969} computed its size. In our simulation with six equally accurate
models examined $200$ times, fixed-sample rank sets at nominal level $5\%$ make at least one
false statement in $30\%$ to $51\%$ of runs (Section~\ref{sec:experiments}, experiment E1). The
fixed-sample guarantee is not wrong: it covers one look, not the way leaderboards are used.

\paragraph{What is needed.}
A leaderboard that is monitored needs rank intervals that are \emph{anytime-valid}. Fix a level
$\alpha\in(0,1)$ in advance. A procedure is anytime-valid if the probability that it is wrong at
any look, over the whole evaluation, is at most $\alpha$. Equivalently, its error guarantee holds
at all sample sizes simultaneously. It then holds at whatever sample size the evaluation stops,
even when the decision to stop depends on the data seen so far. For a single mean, intervals with
this property are called \emph{confidence sequences} \citep{howard2021,waudbysmith2024}. Two
further features of leaderboards matter. First, the statement must cover the ranks of all models
at once, so the method must correct for making many comparisons. Second, all models are scored on
the same items, so their scores are dependent. An item that is easy for one model tends to be easy
for the others, and the strength of this dependence is unknown.

\paragraph{What exists.}
To our knowledge, no existing method gives anytime-valid confidence sets for the ranks of all
models when the models are scored on the same items. Existing rank intervals are
fixed-sample. They certify pairwise comparisons of the form ``model $j$ is better than model $l$''
with a multiple-testing procedure and convert them into rank intervals
\citep{mogstad2024,almohamad2022,neuhof2026mmlu,neuhof2026rank}. The anytime-valid methods that
exist answer other questions. Some track the accuracy of one model
\citep{zhou2026celeus,hsushekhar2026}, and one compares a single pair of models
\citep{kotawala2026resolution}. \citet{arnold2026smcs} certify which models may be the best, a
statement about rank one only. \citet{gu2026serpant} give anytime-valid rank sets for a different
kind of data, pairwise battles. In a battle, two models answer the same prompt and a judge picks
the better answer, with the pair chosen adaptively. \citet{gu2026serpant} correct for multiplicity by Bonferroni,
testing each of the $\M(\M-1)$ ordered pairs at level $\alpha/\{\M(\M-1)\}$. The logical
relations among the pairs are used only afterwards, to extend certified comparisons by
transitivity. Group-sequential designs with Pocock boundaries \citep{pocock1977,arviv2026stop} stop pairwise comparisons early, but
only at looks scheduled in advance and without a correction across pairs.
A side-by-side comparison of these methods is in Appendix~\ref{app:related}.

\paragraph{What we do.}
We construct \emph{rank confidence sequences}. For every model, a rank confidence sequence gives a
set of ranks that contains the model's true rank. This holds for all models at once and at every
look, with probability at least $1-\alpha$. A model's \emph{ability} is its mean score. Its
\emph{true rank} is one plus the number of models with strictly higher ability
(Section~\ref{sec:setup}). The ability can be the model's average over a fixed benchmark. It can
also be the model's expected score on items drawn from a population. The
guarantee holds in finite samples, with no asymptotic approximation. It requires per-item scores
in $[0,1]$. On a fixed benchmark it also requires that the items be evaluated in a random order,
which the evaluator controls. It makes no assumption about how the scores of different models are
correlated. The method keeps a list of candidate orderings of the models, rankings in which ties are allowed,
and reports only what all surviving orderings agree on. For each pair of models, it measures the
evidence that one scores higher than the other as the wealth of an imaginary gambler betting on
it. It discards an ordering once the average of these wealths over the comparisons the ordering denies has grown large. Section~\ref{sec:method}
gives the details.

\paragraph{Contributions.}
\begin{enumerate}[leftmargin=*,itemsep=1pt]
\item \textbf{A leaderboard can be read at any time without inflating its error.} With
  probability at least $1-\alpha$, all rank statements are correct simultaneously, at every look
  and under any stopping rule (Theorem~\ref{thm:validity}; Section~\ref{sec:experiments}, E1 and E2).
\item \textbf{A leaderboard can state how much of its ranking the data support.} Besides
  certified comparisons and rank sets, the report gives certified top-$k$ sets and tiers. A
  top-$k$ statement certifies that a model is among the $k$ best, or that it is not. Tiers
  $s=1,2,\dots$ group the models by certified comparisons. A model in tier $s$ has at least $s-1$
  models certified above it, so its true rank is $s$ or worse
  (Theorems~\ref{thm:ranks} and~\ref{thm:tiersmain}; Section~\ref{sec:experiments}, E2).
\item \textbf{Evaluation compute is saved.} The evaluation of a model can stop as soon as the
  question asked about it is answered, for example whether it is in the top three. The guarantee
  survives even though the decision to stop depends on the observed scores
  (Proposition~\ref{prop:retire}; Section~\ref{sec:experiments}, E3).
\item \textbf{The price of these guarantees is small.} From the same pairwise evidence, the method
  certifies every comparison that a Bonferroni correction over the pairs certifies, and its exact
  version can certify more (Section~\ref{sec:closed}). At a pre-planned look at halfway, it certifies
  about as many comparisons as the strongest fixed-sample procedure (Section~\ref{sec:experiments}, E4).
\end{enumerate}

\section{Problem setting}
\label{sec:setup}

\paragraph{Models, abilities, ranks.}
Each of $\M\ge2$ models has one scalar ability, its mean benchmark score, and its rank is
one plus the number of models with strictly larger ability. Write $\Mset=\{1,\dots,\M\}$ for the
set of models and $\theta=(\theta_1,\dots,\theta_\M)\in\R^\M$ for the vector of abilities, where
$\R$ is the set of real numbers. Write $\#A$, or $|A|$, for the number of elements of a finite
set $A$. The \emph{true rank} of model $j$ is
\begin{equation}\label{eq:rank}
  \rank_j(\theta)=1+\#\{l\in\Mset:\theta_l>\theta_j\},
\end{equation}
abbreviated $\rank_j$ when $\theta$ is clear from context. Rank $1$ is best, and tied models
share a rank. Two models tied for first both have rank $1$. Let $\pairs=\{(j,l)\in\Mset^2:j\ne l\}$ be the ordered pairs.
The \emph{null hypothesis} for $(j,l)\in\pairs$ is $\Hyp_{jl}:\theta_j\le\theta_l$, that $j$ is not
better than $l$. Rejecting $\Hyp_{jl}$ certifies $\theta_j>\theta_l$, written $j\succ l$.

\paragraph{Data.}
All models are scored on the same items, and we assume nothing about how their scores are
correlated. Item $i$ yields the score vector $x_i=(x_{i1},\dots,x_{i\M})\in[0,1]^\M$, one score
per model. The most common case is binary correctness, where $x_{ij}=1$ if model $j$ answers item $i$
correctly and $x_{ij}=0$ otherwise. For binary scores, $\theta_j$ is the model's accuracy. More
generally, scores may take any value in $[0,1]$.

\paragraph{Two sampling models.}
Our main setting is a fixed benchmark evaluated in random order, and the guarantee holds equally
when the items are a sample from a population. Time $t=0,1,2,\dots$ counts the items revealed so
far, and $X_t=(X_{t1},\dots,X_{t\M})$ is the score vector of the $t$th item revealed.
\begin{description}[leftmargin=1.5em,itemsep=1pt]
\item[(F) Finite benchmark.] The benchmark consists of $N$ items and is a fixed matrix
  $x\in[0,1]^{N\times\M}$, where $x_{ij}$ is the score of model $j$ on item $i$. The ability
  $\theta_j=N^{-1}\sum_{i=1}^N x_{ij}$ is the mean score of model $j$ on the whole benchmark.
  One permutation $\Pi$ of $[N]=\{1,\dots,N\}$ is drawn uniformly at random and shared by all
  models. At time $t\le N$, item $\Pi(t)$ is revealed. Its score vector is $X_t=x_{\Pi(t)}$, the
  row $\Pi(t)$ of $x$. The only randomness in the data is the order of evaluation.
\item[(S) Superpopulation.] $X_1,X_2,\dots$ are independent and identically distributed (i.i.d.) draws from a distribution on $[0,1]^\M$,
  and $\theta_j=\E[X_{1j}]$, where $\E$ denotes expectation. The joint law of $X_1$ across models
  is unrestricted.
\end{description}
Under both sampling models, $\F_t$ denotes
the information available after $t$ items. Formally, write $\sigma(\cdot)$ for the
$\sigma$-algebra generated by the listed random variables. Then
$\F_t=\sigma(\Pi(1),\dots,\Pi(t),\xi)$ under (F) and $\F_t=\sigma(X_1,\dots,X_t,\xi)$ under (S). Here
$\xi$ is an auxiliary random variable that lets the analyst randomize decisions. It is independent of $\Pi$ under (F) and of $(X_1,X_2,\dots)$
under (S). The increasing family $(\F_t)_{t\ge0}$ is called the filtration. A quantity is
\emph{adapted} if its value at time $t$ is determined by $\F_t$. Equivalently, it can be computed
from what has been observed by time $t$.

\paragraph{Goal.}
We want rank sets that cover the true ranks of all models at once, at every time, not at one
pre-chosen sample size. Fix a level $\alpha\in(0,1)$ and write $\Pr_\theta$ for probability under a
data-generating law whose ability vector is $\theta$. A \emph{rank confidence sequence} at level
$\alpha$ is a family of adapted random sets $\Rset_{j,t}\subseteq\Mset$, one per model $j$ and
time $t$, with $\Pr_\theta(\forall t\ge0,\ \forall j\in\Mset:\ \rank_j(\theta)\in\Rset_{j,t})\ge1-\alpha$
for every data-generating law.

\section{Rank confidence sequences}
\label{sec:method}

\paragraph{Overview and guarantee.}
After every item, the method keeps the orderings of the models, ties allowed, that the data have
not yet contradicted. It reports only statements on which all surviving orderings agree. It repeats
three steps after each item $t=1,2,\dots$. Step~1 \emph{updates} a nonnegative measure of
evidence for each ordered pair of models. Step~2 \emph{eliminates} every ordering against which
enough evidence has accumulated. Step~3 \emph{reports} the certified dominances $j\succ l$, rank sets, top-$k$ sets
and tiers on which the remaining orderings agree. An optional Step~4 \emph{retires} models whose
evaluation can stop. With probability at least $1-\alpha$, the true ordering, the one given by
the true abilities, is never eliminated. Every report is
agreed on by all surviving orderings, the true one included, so with the same probability every
statement in every report is true at all times simultaneously. The user may therefore look after
every item and stop by any rule.

\paragraph{Inputs.}
The method needs the scores $X_{tj}\in[0,1]$ of every model $j\in\Mset$ on a common sequence of
items $t=1,2,\dots$, and two choices made before any data are seen: the level $\alpha\in(0,1)$
and a finite grid $\Lambda\subset(0,1)$ of betting fractions (Step~1), by default
$\Lambda=\{0.03,0.06,0.12,0.25,0.5\}$. The grid affects how fast evidence accumulates and never
the guarantee. Under (F), a benchmark of $N$ items, the user also draws
one uniformly random order of the items from a recorded seed before the evaluation, and evaluates
every model in that order. Then $X_t=(X_{t1},\dots,X_{t\M})$ is the score vector of the $t$th
item in that order. Under (S), the item stream is open-ended and i.i.d., so we set $N=\infty$.

\paragraph{Step 1: update the pairwise wealths.}
For each ordered pair $(j,l)\in\pairs$ the method maintains a \emph{wealth} $\Ep^{jl}_t\ge0$.
After $t$ items it measures the evidence that model $j$ is better than model $l$, or equivalently
the evidence against $\Hyp_{jl}:\theta_j\le\theta_l$. It is the average capital of $|\Lambda|$ gamblers,
each starting with one unit and, on every item, staking a fixed fraction $\lambda\in\Lambda$ of their
capital on model $j$ outscoring model $l$. Formally, for each $\lambda\in\Lambda$ put $\Ep^{jl}_{\lambda,0}=1$ and, on
item $t$ with paired difference $Z^{jl}_t=X_{tj}-X_{tl}$ and offset $b^{jl}_t$ as given below,
\begin{equation}\label{eq:update}
  \Ep^{jl}_{\lambda,t}=\Ep^{jl}_{\lambda,t-1}
    \Bigl(1+\lambda\,\frac{Z^{jl}_t-b^{jl}_t}{1+b^{jl}_t}\Bigr),
  \qquad
  \Ep^{jl}_t=\frac{1}{|\Lambda|}\sum_{\lambda\in\Lambda}\Ep^{jl}_{\lambda,t},
\end{equation}
where the \emph{offset} $b^{jl}_t$, computed from items $1,\dots,t-1$ only, bounds the
conditional mean of $Z^{jl}_t$ under $\Hyp_{jl}$:
\begin{equation}\label{eq:offset}
  b^{jl}_t=\begin{cases}
    0 & \text{under (S)},\\[2pt]
    \max\Bigl\{-0.99,\ \min\Bigl\{1,\ \dfrac{-S^{jl}_{t-1}}{N-t+1}\Bigr\}\Bigr\} & \text{under (F)},
  \end{cases}
  \qquad S^{jl}_{t-1}=\sum_{s<t}Z^{jl}_s .
\end{equation}
Here $S^{jl}_{t-1}$ is the cumulative paired difference over the first $t-1$ items. Under
$\Hyp_{jl}$, $\E[Z^{jl}_t\mid\F_{t-1}]\le b^{jl}_t$. For i.i.d.\ items this follows from
$\E[Z^{jl}_t]=\theta_j-\theta_l\le0$. On a finite benchmark the next item is drawn uniformly from
the $N-t+1$ items not yet revealed, whose differences sum to
$N(\theta_j-\theta_l)-S^{jl}_{t-1}\le-S^{jl}_{t-1}$. The truncation to $[-0.99,1]$ keeps each
factor in \eqref{eq:update} nonnegative and well defined. Thus, under $\Hyp_{jl}$, each factor
has conditional expectation at most one, so the bets are fair or unfavorable and the wealth
cannot grow in expectation. By Ville's inequality \citep{ville1939}, its probability of ever
reaching $1/\alpha$ is at most $\alpha$. Such a wealth is an \emph{e-process} \citep{ramdas2023},
a nonnegative process whose expected value stays at most one under any stopping rule. Its
reciprocal acts as a p-value that remains valid when it is checked after every item.
When $\theta_j>\theta_l$, the wealth grows exponentially if, for some $\lambda\in\Lambda$, the
factor in \eqref{eq:update} has a positive expected logarithm. The best bet size depends on the
unknown gap $\theta_j-\theta_l$, and averaging over a grid avoids relying on a single one. The
pairs $(j,l)$ and $(l,j)$ have separate wealths.

\paragraph{Step 2: eliminate orderings and certify dominances.}
An ordering of the models is eliminated once the wealths that should stay small if that ordering
were true average to $1/\alpha$ or more. The statement ``$j$ is better than $l$'' is certified
once every ordering that says otherwise has been eliminated. This is the principle of
\emph{closed testing} \citep{marcus1976}, with one test per ordering. A \emph{weak order} $\wo$ on $\Mset$
is a ranking of the models in which ties are allowed. Write $j\wlt_\wo l$ if $\wo$ ranks $l$
strictly above $j$, and $j\weq_\wo l$ if it ties them. Write $j\wle_\wo l$ when $\wo$ ranks $l$ at
least as high as $j$, including a tie. Formally $\wle_\wo$ is a reflexive, transitive and total
relation on $\Mset$. We represent $\wo$ by a vector of levels $(v_1,\dots,v_\M)$, with
$j\wle_\wo l$ if and only if $v_j\le v_l$. Let $\W$ be the set of all weak orders on $\Mset$. It
has $13$ elements for $\M=3$ and $545{,}835$ for $\M=8$. The abilities determine one of them, the \emph{true ordering} $\wo(\theta)$. It satisfies
$j\wle_{\wo(\theta)}l$ if and only if $\theta_j\le\theta_l$. The rank \eqref{eq:rank} depends on
$\theta$ only through the true ordering. Under $\wo$ the rank of model $j$ is
$\rank_j(\wo)=1+\#\{l\in\Mset:j\wlt_\wo l\}$. If $\wo$ were the true ordering, the hypotheses
$\Hyp_{jl}$ that hold would be exactly those of the pairs in
$\T(\wo)=\{(j,l)\in\pairs:\ j\wle_\wo l\}$. By Step~1, none of the wealths of these pairs would be
likely to become large. The evidence against $\wo$ is therefore the average of these wealths. Let
$\Ep^\wo_t$ denote that average. We eliminate $\wo$ permanently when $\Ep^\wo_t$ reaches $1/\alpha$, and
let $\CS_t$ be the set of surviving orderings:
\begin{equation}\label{eq:EW}
  \Ep^\wo_t=\frac{1}{|\T(\wo)|}\sum_{(j,l)\in\T(\wo)}\Ep^{jl}_t,
  \qquad
  \CS_t=\Bigl\{\wo\in\W:\ \max_{s\le t}\Ep^\wo_s<\frac1\alpha\Bigr\}.
\end{equation}
An average of e-processes on the same items is again an e-process, whatever their dependence,
as for e-values \citep{vovk2021}. The set $\CS_t$ of surviving orderings is therefore a confidence set for the
true ordering that is valid at all times at once, $\Pr_\theta(\forall t\ge0:\ \wo(\theta)\in\CS_t)\ge1-\alpha$
(Theorem~\ref{thm:validity}). The certified dominances and the rank sets are what all surviving
orderings agree on,
\begin{equation}\label{eq:D}
  \D_t=\bigl\{(j,l)\in\pairs:\ l\wlt_\wo j\ \text{for all }\wo\in\CS_t\bigr\},
  \qquad
  \Rset_{j,t}=\bigl\{\rank_j(\wo):\ \wo\in\CS_t\bigr\}.
\end{equation}
Because elimination is permanent, $\CS_t$ can only shrink and $\D_t$ can
only grow as $t$ increases.

\paragraph{Computing Step 2.}
Step~2 can be computed in three ways. Algorithm~\ref{alg:exact} enumerates the surviving weak
orders. It alone gives the exact rank sets and is practical up to about $\M=8$.
Algorithm~\ref{alg:shortcut} is a polynomial-time shortcut for any $\M$ that may certify fewer
comparisons. Algorithm~\ref{alg:ilp} computes the same $\D_t$ by integer programming at larger
$\M$. It solved every instance in our experiments. The shortcut certifies $(j,l)$ once the pooled
statistic $B^{jl}_t=\Ep^{jl}_t+\sum_{m\in\Mset\setminus\{j,l\}}\min\{\Ep^{jm}_t,\Ep^{ml}_t\}$ reaches
$\M(\M-1)/\alpha$. Each term $\min\{\Ep^{jm}_t,\Ep^{ml}_t\}$ is indirect evidence that $\theta_j>\theta_l$
through a third model $m$. It is large only if $\theta_j>\theta_m$ and $\theta_m>\theta_l$ are both well supported.
The shortcut certifies
a subset of $\D_t$ and at least what a Bonferroni correction over the $\M(\M-1)$ pairs certifies,
namely every pair whose wealth has reached $\M(\M-1)/\alpha$, a rule we call \emph{e-Bonferroni}.
The \emph{exact test}, computed by Algorithm~\ref{alg:exact} or~\ref{alg:ilp}, can certify more, because testing whole orderings uses the logical links among the
pairwise hypotheses. For example, $\Hyp_{jl}$ and $\Hyp_{lm}$ together imply $\Hyp_{jm}$, so
combinations of true and false hypotheses that no ordering allows need not be guarded against.

\paragraph{Step 3: report.}
Except for the exact rank sets obtained by enumeration, every reported quantity is a function of
the certified set $\D_t$, so all of them are correct whenever $\D_t$ contains no
false dominance. The subscript $t$ is dropped below when
the time is fixed.
\begin{description}[leftmargin=1.5em,itemsep=1pt,topsep=2pt]
\item[Dominances.] Each $(j,l)\in\D_t$ is reported as ``$j$ is better than $l$'', written
  $j\succ l$.
\item[Ranks.] The rank interval of model $j$ is $[L_{j,t},U_{j,t}]$ with
  \begin{equation}\label{eq:LU}
    L_{j,t}=1+\#\{l:(l,j)\in\D_t\},\qquad U_{j,t}=\M-\#\{l:(j,l)\in\D_t\}.
  \end{equation}
  Enumeration can instead report the \emph{exact rank set}
  $\Rset_{j,t}\subseteq[L_{j,t},U_{j,t}]$, which may be smaller.
\item[Top-$k$ sets.] For $k\in\Mset$, model $j$ is certified inside the top $k$ if $U_j\le k$ and
  outside it if $L_j>k$. Ties can put more than $k$ models inside.
\item[Tiers.] A \emph{chain above $j$} of length $h$ is a sequence of models $l_1,\dots,l_h$ such
  that $(l_1,l_2)$, \dots, $(l_{h-1},l_h)$ and $(l_h,j)$ all belong to $\D$. Let $\hgt(j)$ be the greatest such length, with
  $\hgt(j)=0$ if no model is certified better than $j$. The \emph{tier} of $j$ is
  $\tier(j)=1+\hgt(j)$, and tier $s$ is the set $\{j:\tier(j)=s\}$. Tiers are computed by the
  recursion $\tier(j)=1+\max\{\tier(l):(l,j)\in\D\}$, with $\tier(j)=1$ if the set is empty. A
  model in tier $s$ has true rank $s$ or worse.
\end{description}
If $\D$ ever contains a cycle, such as $j\succ l$ and $l\succ j$, or no ordering survives, an error has occurred, which happens with
probability at most $\alpha$. Report it as such.

\paragraph{Step 4 (optional): retire models.}
A model may be dropped from further evaluation at any time, by any rule that uses only the
scores observed so far, and the guarantee still holds. Let
$A_t\subseteq\Mset$ be the set of models scored on item $t$, with $A_{t+1}\subseteq A_t$, chosen
from the scores observed on items $1,\dots,t-1$. The wealth of $(j,l)$ is updated by
\eqref{eq:update} while both $j$ and $l$ belong to $A_t$. Afterwards it is frozen,
$\Ep^{jl}_t=\Ep^{jl}_{t-1}$, and continues to be used in Steps~2 and~3. Natural rules retire $j$
once its top-$k$ status is certified, once every pair involving $j$ is certified in one
direction, or once $U_j-L_j$ is below a target. The item order is not affected by retirement, and a retired model is not brought back.

\paragraph{Scope and cost.}
The model set must be fixed in advance, items must not be chosen adaptively, and every model
that has not been retired must have a score on each item. Pairwise-battle data are out of scope.
The per-item update is negligible next to running the models. It costs $O(\M^2)$ per item,
$25\,\mu$s for $20$ models on one CPU core. Appendix~\ref{app:method} states
Algorithms~\ref{alg:exact}--\ref{alg:ilp} with implementation details and costs. Appendix~\ref{app:rigor} gives the limits of the guarantee.

\section{What can be established, and why it matters}
\label{sec:theory}

The guarantee of Section~\ref{sec:method} follows from three results, one per subsection, all proved in Appendix~\ref{app:proofs}.

\subsection{All certified dominances are correct, at all times}
\label{sec:closed}

\paragraph{What needs a guarantee.}
Every report except the exact rank sets is computed from the certified set $\D_t$, so the central requirement is that, with
probability at least $1-\alpha$, no $\D_t$ at any time $t$ contains a false dominance, a pair
$(j,l)$ with $\theta_j\le\theta_l$. A pair with $\theta_j>\theta_l$ is a true dominance. We can
guarantee this over all pairs and times at once, even though the wealths share items and may be
dependent in unknown ways. The guarantee starts from one property of the pairwise wealths, that of an e-process
\citep{ramdas2023}. Identify a hypothesis $H$ with the set of data-generating laws it allows and
let $\E_P$ denote expectation under the law $P$. A \emph{stopping time} is a random time $\tau$
such that whether $\tau=t$ can be decided from the information $\F_t$ available at time $t$. Formally, an
\emph{e-process} for $H$ is a nonnegative adapted process $(\Ep_t)_{t\ge0}$ with
$\E_P[\Ep_\tau]\le1$ for every $P\in H$ and every bounded stopping time $\tau$. An e-process satisfies Ville's inequality:
$P(\exists t\ge0:\ \Ep_t\ge1/\alpha)\le\alpha$ for every $P\in H$, where $P(\cdot)$ is
probability under the law $P$. This is the counterpart of Markov's inequality that holds over
all times at once. Under both sampling models every wealth of
Section~\ref{sec:method} is an e-process for its hypothesis $\Hyp_{jl}$, whatever the dependence
between models. This is the betting construction of
\citet{waudbysmith2024}. Because
all models share one item order, all wealths are e-processes for the \emph{same} filtration. We
record this as an assumption.
\begin{assumption}\label{ass:eproc}
There is a filtration $(\F_t)$ and, for each $(j,l)\in\pairs$, a nonnegative adapted process
$(\Ep^{jl}_t)$ with $\Ep^{jl}_0=1$ that is an e-process for $\Hyp_{jl}$ with respect to $(\F_t)$
under every data-generating law for which $\theta_j\le\theta_l$. No condition is placed on the
joint behavior of the $\M(\M-1)$ processes.
\end{assumption}

\paragraph{Result.}
With probability at least $1-\alpha$ the true ordering is never eliminated, and as long as it
survives no certified dominance is false. Write $\thetastar$ for the true ability vector and
$\wostar=\wo(\thetastar)$ for the true ordering. Here $\CS_t$ and $\D_t$ are the surviving and
certified sets of \eqref{eq:EW} and \eqref{eq:D}.

\begin{theorem}[Validity]\label{thm:validity}
Under Assumption~\ref{ass:eproc}, for every data-generating law with parameter $\thetastar$ and
$\wostar=\wo(\thetastar)$, and with $\Pr$ denoting probability under that law,
\begin{equation}\label{eq:cover}
  \Pr\bigl(\forall t\ge0:\ \wostar\in\CS_t\bigr)\ge1-\alpha .
\end{equation}
Consequently the family-wise error rate, the probability of certifying any false dominance, is
controlled uniformly over time,
\begin{equation}\label{eq:fwer}
  \Pr\bigl(\exists t\ge0,\ \exists(j,l)\in\D_t:\ \thetastar_j\le\thetastar_l\bigr)\le\alpha .
\end{equation}
Both statements also hold when $t$ is replaced by any $\N_0$-valued random time, stopping or
otherwise, where $\N_0=\{0,1,2,\ldots\}$.
\end{theorem}

The proof applies Ville's inequality to $\Ep^{\wostar}_t$, an average of wealths of true
hypotheses. While $\wostar$ survives, the certified set is a strict partial order (no pair in both directions, closed under
transitivity), so the report never contradicts itself.

\paragraph{Relation to existing results.}
Theorem~\ref{thm:validity} combines closed testing over weak orders with averaging of
e-processes. We use the partitioning principle of
\citet{finner2002}, where the orderings partition the parameter space and one test is carried out
per cell, as in the fixed-sample rank intervals of \citet{almohamad2021}. \citet{shaffer1986} shows that only sets of hypotheses that can be true together
need be tested. Here these are exactly the sets $\T(\wo)$, one per weak order. Averaging
e-values (nonnegative statistics with expectation at most one under the null), the fixed-time
counterpart of e-processes, remains valid under arbitrary dependence
\citep{vovk2021,wang2025}, while products generally require independence, which shared items rule
out. Our result differs from the two closest ones as follows. \citet{arnold2026smcs} average
e-processes and close over all subsets of models, which certifies the set of models that may be
best and no ranks. \citet{hartoglei2025} close over all sets of hypotheses, the full closure, which
controls the family-wise error for arbitrary hypotheses but yields no rank statements. We test one averaged
e-process per weak order, which yields rank sets. Nothing is lost relative to standard
corrections. On the same wealths, the exact test certifies every pair that the full closure
certifies, and the shortcut every e-Bonferroni rejection used by \citet{gu2026serpant}.
Appendix~\ref{sec:power} states both results.

\subsection{Ranks, top-\texorpdfstring{$k$}{k} sets and tiers inherit the guarantee}
\label{sec:duality}

\paragraph{What needs a guarantee.}
Users read ranks, top-$k$ sets and tiers, which are built from the dominances. These reported quantities must also be
correct without costing additional error. To cover the exact test, the
shortcut and
e-Bonferroni at once, call a family $(\D_t)_{t\ge0}$ of random subsets of $\pairs$ a
\emph{certified set} if its \emph{good event}
$\good=\{\forall t\ge0,\ \forall(j,l)\in\D_t:\ \thetastar_j>\thetastar_l\}$ has probability at
least $1-\alpha$. All three are certified sets. Taking the transitive closure adds no error. Let
$\Dbar_t$ be the transitive closure of $\D_t$, which adds $(j,m)$ whenever $(j,l)$ and $(l,m)$
are present, repeatedly. On $\good$, the transitive closure contains only
true dominances. True dominances form a strict partial order, so the closure is certified on the
same event.

\paragraph{Result.}
Rank intervals, top-$k$ sets and tiers are correct on $\good$, and exact rank sets while
$\wostar$ survives.

\begin{theorem}[Rank-set duality]\label{thm:ranks}
(a) \emph{Projection.} Under Assumption~\ref{ass:eproc}, the rank sets $\Rset_{j,t}$ of
\eqref{eq:D} satisfy $\Pr(\forall t,\forall j:\ \rank_j(\thetastar)\in\Rset_{j,t})\ge1-\alpha$.
(b) \emph{Counting.} Let $(\D_t)$ be any certified set and compute $L_{j,t}$ and $U_{j,t}$ from
it by \eqref{eq:LU}. On $\good$, $L_{j,t}\le\rank_j(\thetastar)\le U_{j,t}$ for all $j$ and all
$t$, and the intervals only improve if $\D_t$ is replaced by $\Dbar_t$.
(c) \emph{Comparison.} For $\D_t$ from \eqref{eq:D}, $\Rset_{j,t}\subseteq[L_{j,t},U_{j,t}]$,
and the inclusion is strict for some wealth configurations.
\end{theorem}

By part (b), on $\good$ every top-$k$ certificate of Step~3 is correct, for all $j$, $k$ and $t$
simultaneously.

\begin{theorem}[Tier guarantee]\label{thm:tiersmain}
Let $(\D_t)$ be a certified set and let $\tier_t(j)$ be the tier of model $j$ computed from
$\Dbar_t$ as in Step~3 of Section~\ref{sec:method}. On $\good$, simultaneously for all $j$ and
all $t$,
\begin{equation}\label{eq:tierbound}
  \rank_j(\thetastar)\ \ge\ L_{j,t}\ \ge\ \tier_t(j),
\end{equation}
where $L_{j,t}$ is computed from $\Dbar_t$, and the number of tiers is at most the number of
distinct values among $\thetastar_1,\dots,\thetastar_\M$. On $\good$, if
$\D_t\subseteq\D_{t+1}$ then $\tier_t(j)\le\tier_{t+1}(j)$, so a model's tier index never decreases
as evidence accrues.
\end{theorem}

The guarantee is only \eqref{eq:tierbound}. A lower
tier index does not imply a certified pairwise dominance, and tier $1$ also holds every model
with no certified comparison.

\paragraph{Relation to existing results.}
The counting part of Theorem~\ref{thm:ranks} comes from \citet{almohamad2022} and
\citet{mogstad2024}, and \citet{gu2026serpant} apply it at all times. They also define the same
tiers without a guarantee, which Theorem~\ref{thm:tiersmain} adds. Appendix~\ref{app:tiers} gives more on tiers.

\subsection{Evaluation can stop early, model by model}
\label{sec:stopping}

\paragraph{What needs a guarantee.}
Retiring a model, that is, dropping it from further evaluation, is data-dependent. Once a model
stops, its future scores are not observed and every pair involving it has its wealth frozen. The
guarantee must remain valid under both effects for any rule based only on observed scores.
After retirement the analyst observes less information than $\F_t$, which contains the scores of
all models on all revealed items.

\paragraph{Result.}
Freezing each wealth at the retirement of either of its models preserves
Assumption~\ref{ass:eproc}, and with it every result that rests on that assumption.

\begin{proposition}[Retirement preserves validity]\label{prop:retire}
Under \textup{(F)} or \textup{(S)}, let the item order in \textup{(F)} be drawn before evaluation
and remain unaffected by later decisions. At each time $t$ let $A_t\subseteq\Mset$ be the set of
models scored on item $t$, chosen as a function of the scores observed on items $1,\dots,t-1$ and
of $\xi$, with $A_1=\Mset$, $A_{t+1}\subseteq A_t$, and $A_{N+1}=\emptyset$ under \textup{(F)}.
Model $j$ is retired at $\sigma_j=\sup\{t:j\in A_t\}$ ($\sigma_j=\infty$ if never), so that item
$\sigma_j$ is included in its record, and retirement is permanent under both sampling models. Let
$\ind{\cdot}$ denote the indicator of the event in braces and $s\wedge t=\min\{s,t\}$, and let
$\F^{\mathrm{obs}}_t=\sigma\bigl(\xi,\ \ind{j\in A_s},\ \ind{j\in A_s}X_{sj}:\ s\le t,\ j\in\Mset\bigr)$
be the observed-data filtration, generated by which models were scored and the scores observed.
Call a bet
\emph{predictable} for a filtration if the bet on item $t$ is determined by the information at
time $t-1$. Then $\F^{\mathrm{obs}}_t\subseteq\F_t$, and each $\sigma_j$ is an
$(\F^{\mathrm{obs}}_t)$-stopping time and hence an $(\F_t)$-stopping time. Replacing each
$\Ep^{jl}_t$ by $\Ep^{jl}_{t\wedge\sigma_j\wedge\sigma_l}$, with any
$(\F^{\mathrm{obs}}_t)$-predictable bets, preserves
Assumption~\ref{ass:eproc} and every result resting on it.
\end{proposition}

\paragraph{Relation to existing results.}
Retirement times and bets chosen from the observed scores alone are also stopping times and predictable bets for $(\F_t)$. All frozen wealths therefore remain e-processes for one common
filtration, while the procedure never needs an unobserved score. Freezing at such times matters.
\citet{tavyrikov2025} show that averaging running maxima, which freezes each e-process at its
best past value, need not control the family-wise error. Existing sequential evaluation
methods stop evaluating one model \citep{zhou2026celeus,hsushekhar2026} or one pair at a time
\citep{kotawala2026resolution,gu2026serpant,arviv2026stop}. Ours stops models one by one inside a
simultaneous guarantee over all ranks. Appendix~\ref{app:rigor} says when a frozen pair may be updated again.

\section{Experiments}
\label{sec:experiments}

Four experiments, E1 to E4, test monitoring, certification on real data, compute savings and
the price of anytime validity. The real-data experiments use per-item results
of $395$ large language models (LLMs) from the Open LLM Leaderboard \citep{maiapolo2024tiny,openllmleaderboard2023}. Setups
are in Appendix~\ref{app:experiments}, with further experiments on cross-model dependence and on
what the multiplicity correction and the bet contribute.

\paragraph{E1. Only fixed-sample rank sets break under monitoring.}
With six tied models checked at $200$ looks and nominal level $5\%$, fixed-sample rank sets err in
$30\%$ to $51\%$ of runs and our exact test in $1.6\%$ (Figure~\ref{fig:peeking}a). The models answer
$2000$ simulated i.i.d.\ items, with equally spaced looks. The fixed-sample rank sets are the standard construction, Holm's correction \citep{holm1979} over
paired tests of all $\M(\M-1)$ ordered pairs, converted into rank intervals by \eqref{eq:LU}.
They err in $30\%$ of runs with the exact McNemar test and $51\%$ with the paired $z$-test. With
tied pairs at ranks 1 and 3 the rates are $9.7\%$, $18\%$ and ours $0.3\%$. With no exact ties
but gaps of $0.004$, they are $4.6\%$, $11.7\%$ and $0.2\%$. At a single pre-planned look both
fixed-sample procedures are valid, at $3.3\%$ and $3.6\%$ with all six tied, so the inflation is
caused by monitoring alone and not by the normal approximation.

\begin{figure}[!t]
\centering
\includegraphics[width=\textwidth]{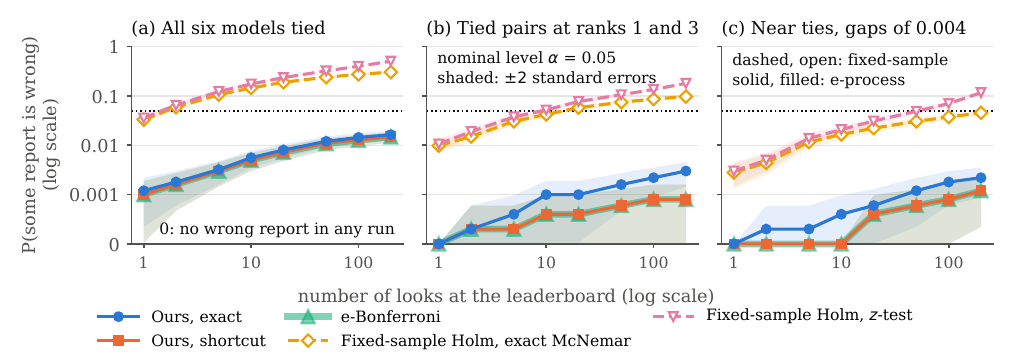}
\caption{Probability that some report is wrong against the number of looks, $\alpha=0.05$.
Most fixed-sample curves (dashed) exceed $\alpha$ as looks accumulate. E-process curves (solid)
stay below $\alpha$. Both axes use a log scale: tick labels are actual values a factor of ten apart,
except 0.}
\label{fig:peeking}
\end{figure}

\paragraph{E2. The data support only part of the ranking early.}
Figure~\ref{fig:real} follows the twenty highest-scoring models across
six benchmarks, while Table~\ref{tab:scale} summarizes the halfway and final results. On MMLU,
halfway through the benchmark, $50\%$ of the $190$ true dominances, the pairs whose order
on the full benchmark is strict, are certified and the
twenty models form $3.0$ certified tiers. HellaSwag is similar, while the other four benchmarks
separate fewer pairs at halfway. The difficulty comes from choosing twenty models near the top.
When twenty models are spread across the leaderboard, $91\%$ of their true dominances on MMLU are already
certified halfway through. We also used the top eight models. Of all $900$ runs, one had a wrong report, far below the $5\%$ level.

\paragraph{E3. Model-by-model stopping saves evaluation.}
The top twenty are chosen with the full benchmark in hand, which an operator cannot do, so we
repeat on all $395$ released models with no further selection. Most are far from the top-$3$ boundary, so their
evaluation can stop early. Certifying top-$3$ membership for all $395$ uses $6.2\%$ of the full
MMLU evaluation, against $31\%$ for the twenty highest-scoring alone. In absolute terms the $395$-model
run costs about four times as much. In simulation, retiring each model once its own top-$3$ question is
answered costs $22\%$ against $84\%$ for stopping every model at the same moment, and the cost
barely depends on $k$. The full leaderboard is also better resolved, with $86\%$ of the true dominances on MMLU
certified halfway and $99.6\%$ by the end (Table~\ref{tab:scale}). No report was wrong in any of these runs.

\begin{table}[t]
\caption{True dominances certified, as a percentage of those that exist, and mean tiers at halfway.
Ours checked at every $1\%$, $50$ item orders. No report was wrong in any of these runs.}
\label{tab:scale}
\centering\small
\begin{tabular}{@{}lrrrcrrr@{}}
\toprule
& \multicolumn{3}{c}{Top twenty ($179$--$190$ dominances)} & & \multicolumn{3}{c}{All $395$ models ($\approx77{,}700$)}\\
\cmidrule(lr){2-4}\cmidrule(lr){6-8}
Benchmark & half & end & tiers & & half & end & tiers\\
\midrule
MMLU          & $50\%$ & $99\%$ & $3.0$ & & $86\%$ & $99.6\%$ & $21.2$\\
HellaSwag     & $49\%$ & $96\%$ & $3.2$ & & $77\%$ & $99.2\%$ & $20.4$\\
GSM8K         & $22\%$ & $96\%$ & $2.9$ & & $78\%$ & $98.9\%$ & $10.8$\\
WinoGrande    & $2\%$  & $85\%$ & $2.0$ & & $26\%$ & $95.5\%$ & $4.2$\\
ARC-Challenge & $0\%$  & $82\%$ & $1.1$ & & $48\%$ & $97.0\%$ & $6.8$\\
TruthfulQA    & $17\%$ & $82\%$ & $2.0$ & & $55\%$ & $97.0\%$ & $6.1$\\
\bottomrule
\end{tabular}
\end{table}

\begin{figure}[!t]
\centering
\includegraphics[width=0.97\textwidth]{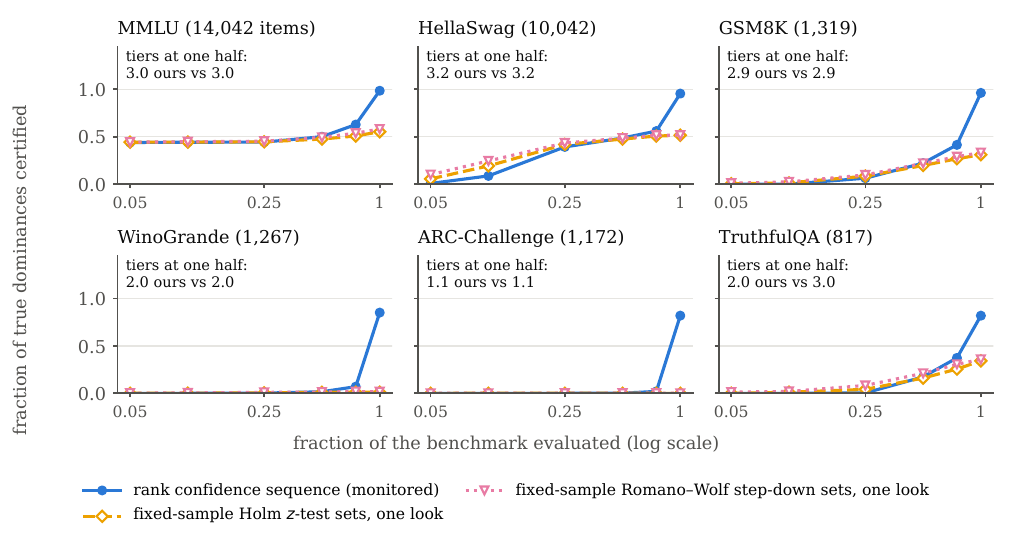}
\caption{Twenty highest-scoring models, $50$ item orders: fraction of true dominances
certified against fraction evaluated. Ours (Algorithm~\ref{alg:shortcut}) checked at
every $1\%$, Holm and Romano--Wolf at one look. Each panel gives the mean tiers at halfway, ours
against Romano--Wolf.}
\label{fig:real}
\end{figure}

\paragraph{E4. Anytime validity costs little at a pre-planned look.}
At a pre-planned look at halfway, ours matches or exceeds the strongest fixed-sample procedure,
the bootstrap step-down rank sets of \citet{mogstad2024}, based on \citet{romano2005}, on five of the
six benchmarks and trails by four percentage points on TruthfulQA (Figure~\ref{fig:real}).
Beyond halfway the offset of (F), which uses the finiteness of the benchmark, increasingly favors
ours, so the later separation should not be read as a pure power comparison.

\section{Limitations and conclusion}
\label{sec:limits}

\paragraph{Limitations.}
At an early pre-planned look, fixed-sample procedures certify more. Models arriving over time
\citep{fischerramdas2024} and adaptive item selection \citep{zhou2026celeus} are not covered.
Validity covers the benchmark as given. Contamination and transfer to other tasks lie outside it.

\paragraph{Conclusion.}
Rank confidence sequences give shared-item leaderboards rank statements valid for all models at
all times. In our experiments they stayed valid under monitoring where fixed-sample rank sets did
not, matched or exceeded them at a pre-planned look at halfway on five of six benchmarks, and settled
every model's top-$3$ status with $6.2\%$ of a full MMLU evaluation.
The code is available at \url{https://github.com/HamedKhosravi99/rank-confidence-sequences-supplement}.

\bibliography{references}
\bibliographystyle{iclr2027_conference}

\appendix
\raggedbottom
\addtocontents{toc}{\protect\setcounter{tocdepth}{2}}
\renewcommand{\contentsname}{Appendix contents}
\clearpage
\tableofcontents
\clearpage
\section{Method details}
\label{app:method}

Step~2 of Section~\ref{sec:method} can be carried out three ways, and this appendix gives all
three. Enumeration checks every ordering directly and is practical while the number of orderings
is manageable. Integer programming computes the same answer at any size. A closed-form shortcut
is fast but may certify less. All three are run after every update. If they are run less
often, the running maximum in \eqref{eq:EW} is taken over those runs only, so no more orderings are
eliminated. The certified set is then contained in $\D_t$, each rank set contains $\Rset_{j,t}$,
and the guarantee of Theorem~\ref{thm:validity} still holds. Algorithm~\ref{alg:exact} is the
enumeration.

\begin{algorithm}[H]
\caption{Exact certification by enumeration (small $\M$). Run after every update.}\label{alg:exact}
\begin{algorithmic}[1]
\Require wealths $\Ep^{jl}_t$; level $\alpha$; the set $\CS$ of surviving weak orders (initially $\W$)
\For{each weak order $\wo\in\CS$}
  \State $e\gets\frac{1}{|\T(\wo)|}\sum_{(j,l)\in\T(\wo)}\Ep^{jl}_t$ \Comment{average wealth over the pairs with $j\wle_\wo l$}
  \If{$e\ge1/\alpha$} remove $\wo$ from $\CS$ \Comment{permanent}
  \EndIf
\EndFor
\State $\D\gets\{(j,l)\in\pairs:\ l\wlt_\wo j\text{ for every }\wo\in\CS\}$ \Comment{certified: $j$ better than $l$}
\State $\Rset_j\gets\{1+\#\{l:j\wlt_\wo l\}:\ \wo\in\CS\}$ for each $j$ \Comment{rank set of model $j$}
\State \Return $\CS,\D,(\Rset_j)_j$
\end{algorithmic}
\end{algorithm}

Beyond about eight models enumeration is impractical. The closed-form rule of
Algorithm~\ref{alg:shortcut} then certifies a subset of $\D_t$ at $O(\M)$ operations per pair, and
Algorithm~\ref{alg:ilp} below computes $\D_t$ exactly. Transitivity supplies the second term of the
shortcut's statistic. Take any ordering that puts $l$ at least as high as $j$, and any third model
$m$. Either the ordering puts $m$ at least as high as $j$, and then the wealth $\Ep^{jm}$ counts
against it, or it puts $m$ below $j$ and hence below $l$, and then $\Ep^{ml}$ counts against it.
Either way the ordering is charged at least $\min\{\Ep^{jm},\Ep^{ml}\}$.

\begin{algorithm}[H]
\caption{Shortcut certification (any $\M$). Run after every update.}\label{alg:shortcut}
\begin{algorithmic}[1]
\Require wealths $\Ep^{jl}$; level $\alpha$; the set $\D$ of certified pairs (initially empty)
\For{each ordered pair $(j,l)\notin\D$}
  \State $B\gets\Ep^{jl}+\sum_{m\ne j,l}\min\{\Ep^{jm},\,\Ep^{ml}\}$
  \If{$B\ge\M(\M-1)/\alpha$} add $(j,l)$ to $\D$ \Comment{permanent}
  \EndIf
\EndFor
\State replace $\D$ by its transitive closure \Comment{if $j\succ m$ and $m\succ l$ then $j\succ l$; same good event, no added error probability}
\State \Return $\D$
\end{algorithmic}
\end{algorithm}

\paragraph{Exact certification for any $\M$.}
Algorithm~\ref{alg:ilp} computes $\D_t$ without enumerating orderings, by asking an integer
program whether any surviving ordering places $l$ at least as high as $j$. A weak order $\wo$ is
encoded by $y\in\{0,1\}^{\pairs}$ with $y_{ab}=1$ exactly when $(a,b)\in\T(\wo)$. The $0/1$
vectors with $y_{ab}+y_{ba}\ge1$ for all $a\ne b$, which is totality, and
$y_{ab}+y_{bc}-y_{ac}\le1$ for all distinct $a,b,c$, which is transitivity, are exactly the weak
orders. A \emph{call time} is a time at which an algorithm is run, and a \emph{wealth matrix} is the array
$(\Ep^{ab}_s)_{(a,b)\in\pairs}$ of pairwise wealths at one call time. With a run after every update
the call times up to $t$ are $1,\dots,t$. The order $\wo$ survives
call time $s$ when $\Ep^\wo_s<1/\alpha$, equivalently when
$g_s(y):=\sum_{(a,b)\in\pairs}y_{ab}(\Ep^{ab}_s-1/\alpha)<0$. For a set
$S$ of call times let
\begin{equation}\label{eq:ilp}
  z^\star(S)=\min\Bigl\{z:\ z\ge g_s(y)\ \text{for all }s\in S,\ y\in\{0,1\}^{\pairs}\text{ a weak order},\ y_{jl}=1\Bigr\},
\end{equation}
an integer program with one continuous variable. Then $(j,l)\in\D_t$ exactly when
$z^\star(\{1,\dots,t\})\ge0$, and a minimizer with $z^\star<0$ is a surviving ordering with
$(j,l)\in\T(\wo)$, a \emph{witness} that $(j,l)\notin\D_t$. Survival needs the strict inequality
$g_s(y)<0$ and certification its complement $z^\star\ge0$, so in exact arithmetic no tolerance is
needed. Algorithm~\ref{alg:ilp} keeps a pool of witnesses and drops one the first time its
average wealth $\Ep^\wo_t$ reaches $1/\alpha$. A witness $\wo$ \emph{covers} every pair in
$\T(\wo)$, and \eqref{eq:ilp} is solved only for pairs that no witness in the pool covers. The
program starts with few call-time constraints and adds them only when needed. This is safe
because a maximum over fewer times is no larger, so $z^\star(S)\ge0$ for any
$S\subseteq\{1,\dots,t\}$ already certifies the pair. Otherwise the minimizer is checked against
every call time, and a time at which it fails is added to $S$.
In exact arithmetic Algorithm~\ref{alg:ilp} computes $\D_t$ (Proposition~\ref{prop:ilp}). Under
numerical optimization the implementation, which solves the programs with HiGHS
\citep{huangfu2018} through SciPy \citep{virtanen2020scipy}, is conservative. It certifies a pair only when the solver's
lower bound on the optimum exceeds a tolerance of $10^{-9}\M(\M-1)/\alpha$, relative to the scale
of $g_s$. It accepts an ordering as a witness
only when it survives every call time by that margin, and it leaves a numerically ambiguous pair
unresolved until the next call. Hence within the solver's numerical guarantees it never
certifies more than $\D_t$. No such unresolved case arose in any run reported here.

\begin{algorithm}[H]
\caption{Exact certification by integer programming (any $\M$). Run after every update.}\label{alg:ilp}
\begin{algorithmic}[1]
\Require wealths $\Ep^{ab}_s$ at every call time $s\le t$; level $\alpha$; the certified set $\D$,
  the pool $\mathcal{Q}$ of witnesses (surviving weak orders found so far) and the set $S$ of
  constraint times, all initially empty
\State remove from $\mathcal{Q}$ every $\wo$ with $\Ep^\wo_t\ge1/\alpha$ \Comment{rejected for good}
\For{each $(j,l)\notin\D$ with $(j,l)\notin\T(\wo)$ for all $\wo\in\mathcal{Q}$, smallest $\Ep^{jl}_t$ first}
  \If{$B^{jl}_t\ge\M(\M-1)/\alpha$} add $(j,l)$ to $\D$ and \textbf{continue} \Comment{Proposition~\ref{prop:shortcut}}
  \EndIf
  \Loop
    \State solve \eqref{eq:ilp} for $S\cup\{t\}$: optimum $z^\star$, minimizer $\wo$
    \If{$z^\star\ge0$} add $(j,l)$ to $\D$ and \textbf{break}
    \EndIf
    \If{$\Ep^\wo_s<1/\alpha$ for every call time $s\le t$} add $\wo$ to $\mathcal{Q}$ and \textbf{break}
    \Else\ add to $S$ a call time $s$ with $\Ep^\wo_s\ge1/\alpha$
    \EndIf
  \EndLoop
\EndFor
\State \Return $\D$
\end{algorithmic}
\end{algorithm}

\begin{proposition}[Algorithm~\ref{alg:ilp} is exact]\label{prop:ilp}
For every $\M\ge2$, every $\alpha\in(0,1)$ and every sequence of wealth
matrices, one per call time, Algorithm~\ref{alg:ilp} in exact arithmetic returns the set $\D_t$
of Algorithm~\ref{alg:exact} at every call time $t$, and each inner loop ends after at most
$t$ solves.
\end{proposition}

Exact certification is hard in the worst case. A decision problem is
\emph{coNP-hard} when every problem whose ``no'' answers can be checked in polynomial time reduces
to it. Such a problem has no polynomial-time algorithm unless $\mathrm{P}=\mathrm{NP}$.

\begin{proposition}[Certification is coNP-hard in the worst case]\label{prop:hard}
Deciding from a wealth matrix and a level $\alpha$ whether a given pair belongs to $\D_1$, the
certified set when the certifier is called once, on that matrix alone, is coNP-hard. The wealth
matrices of the reduction take the values $0$, $2^{\nu-1}$ and $2^\nu$ for an integer $\nu\ge1$,
have polynomial encoding size, and are produced by the general update \eqref{eq:wealth} of
Appendix~\ref{sec:eproc} with $N=\infty$ from polynomially many binary items with predictable
bets in $\{0,1\}$.
\end{proposition}

The reduction is from the feedback arc set problem in tournaments \citep{alon2006,charbit2007}.
A tournament is a directed graph with exactly one arc between every two vertices, and the problem
asks for an ordering of its vertices with the fewest arcs pointing backward. The reduction uses
bets in $\{0,1\}$, and Appendix~\ref{app:rigor} says what remains open for the default mixture
bet. Each integer program \eqref{eq:ilp} has $\M(\M-1)$ binary variables, one continuous variable
and $\M(\M-1)(\M-2)$ transitivity constraints. Algorithm~\ref{alg:ilp} solved every instance in
our experiments, and Appendix~\ref{app:ecorr} reports its running time.

An implementation stores $\log\Ep^{jl}_{\lambda,t}$ instead of $\Ep^{jl}_{\lambda,t}$ to avoid
overflow. As for cost, the update is $O(\M^2)$ per item, Algorithm~\ref{alg:shortcut} is $O(\M^3)$ and
Algorithm~\ref{alg:exact} is $O(\M^2)$ per surviving ordering. In NumPy on one core the update
costs $25\,\mu$s per item for $20$ models and $47\,\mu$s for $50$, the shortcut $0.07$\,ms and
$0.6$\,ms per call, and the enumeration $58$\,ms per call for $8$ models while every ordering
survives, falling as they are discarded. At each call Algorithm~\ref{alg:ilp} solves no integer
program for a pair that is already certified, that the shortcut certifies or that a witness in the
pool covers. For any other pair it solves at least one and at most $t$ (Proposition~\ref{prop:ilp}).

\section{Supporting theory}
\label{app:theory}

In Appendices~\ref{app:theory} to~\ref{app:experiments}, the true ability vector $\thetastar$ is also written $\theta$, as in Sections~\ref{sec:setup} and~\ref{sec:method}.

\subsection{Validity and growth of each pairwise wealth}
\label{sec:eproc}

A large wealth must be improbable when the hypothesis it argues against is in fact true, however
often and whenever the analyst looks. Fix $(j,l)\in\pairs$, write $Z_t=Z^{jl}_t$, and consider
the general form of the update \eqref{eq:update}, $\Ep^{jl}_0=1$ and
\begin{equation}\label{eq:wealth}
  \Ep^{jl}_t=\Ep^{jl}_{t-1}\Bigl(1+\lambda_t\,\frac{Z_t-b_t}{1+b_t}\Bigr),
\end{equation}
in which the offset $b_t\in(-1,1]$ and the bet $\lambda_t\in[0,1]$ are \emph{predictable}, that
is, determined by the information $\F_{t-1}$. Each $\Ep^{jl}_{\lambda,t}$ of \eqref{eq:update}
has this form with $\lambda_t\equiv\lambda$, and so does the grid average $\Ep^{jl}_t$, with
$\lambda_t$ the mean of $\Lambda$ weighted by the wealths $\Ep^{jl}_{\lambda,t-1}$
(Remark~\ref{rem:bets}). The factor in \eqref{eq:wealth} is nonnegative,
and its conditional expectation given $\F_{t-1}$ is at most one whenever
$\E[Z_t\mid\F_{t-1}]\le b_t$. The wealth is then a nonnegative \emph{supermartingale}: given the
past, its expected next value is at most its current value. It is therefore an
e-process for every predictable bet and stays one when frozen at a stopping time. The offsets
\eqref{eq:offset} satisfy that condition for every law with $\theta_j\le\theta_l$, and the
wealth of \eqref{eq:update} is an average over $\lambda\in\Lambda$ of such processes. The choice
of bets therefore affects power and never validity.

The pairwise construction follows the betting approach of \citet{waudbysmith2024}, including,
under (F), its version for sampling without replacement. \citet{waudbysmith2020} built earlier
confidence sequences for that setting without betting.

\begin{lemma}[Pairwise betting supermartingale]\label{lem:betting}
Let $P$ be a law under which $\E_P[Z_t\mid\F_{t-1}]\le b_t$ almost surely for every $t$. Then
$(\Ep^{jl}_t)$ of \eqref{eq:wealth} is a nonnegative $(\F_t)$-supermartingale under $P$ with
$\Ep^{jl}_0=1$, for every predictable choice of $(\lambda_t)$. If the process is frozen at an
$(\F_t)$-stopping time $\sigma$, that is, replaced by $\Ep^{jl}_{t\wedge\sigma}$ where
$t\wedge\sigma=\min\{t,\sigma\}$, the conclusion is unchanged.
\end{lemma}

\begin{proposition}[Both sampling models]\label{prop:instances}
In each case below, for every $(j,l)\in\pairs$, every predictable bet $(\lambda_t)$ with values
in $[0,1]$, and every data-generating law under which $\Hyp_{jl}$ is true, $(\Ep^{jl}_t)$ is a
nonnegative $(\F_t)$-supermartingale started at $1$, hence an e-process for $\Hyp_{jl}$. This
holds whatever the other models' parameters and whatever the dependence between models.
\begin{enumerate}[label=(\alph*),leftmargin=*,itemsep=2pt]
\item Under \textup{(F)}, with $\xi$ independent of $\Pi$, for $t\le N$, with
  \begin{equation}\label{eq:bfinite}
    b_t=\max\Bigl\{-1+\delta,\ \min\Bigl\{1,\ \frac{-S_{t-1}}{N-t+1}\Bigr\}\Bigr\},
    \qquad S_{t-1}=\sum_{s<t}Z_s ,
  \end{equation}
  for any fixed $\delta\in(0,1]$, and $\Ep^{jl}_t=\Ep^{jl}_N$ for $t>N$.
\item Under \textup{(S)}, with $b_t\equiv0$.
\end{enumerate}
\end{proposition}

\begin{remark}[Bets]\label{rem:bets}
Validity holds for every predictable $(\lambda_t)$, so the bet is chosen for power, and a
constant bet is a poor choice when gaps differ across pairs. For binary scores under (S) (the
finite case behaves the same away from the end of the benchmark), let
$\Delta=\theta_j-\theta_l>0$ and let $\rho=\Pr(Z\ne0)$ be the \emph{discordance rate}, the
probability that the two models score differently on an item. Here $Z$ stands for a generic
paired difference $Z^{jl}_t$. The expected log-growth of a constant bet is
$\E\log(1+\lambda Z)\approx\lambda\Delta-\lambda^2\rho/2$: it is maximized at
$\lambda^\star=\Delta/\rho$, which is in fact the exact maximizer when $\Delta<\rho$, and is
\emph{negative} beyond about $2\lambda^\star$. By the strong law of large numbers, a constant
bet with negative expected log-growth drives the wealth of a pair with small $\Delta/\rho$,
close relative to its discordance rate, to zero almost surely. Such a pair may therefore never
be certified, however long the evaluation runs (an early crossing of the threshold remains
possible). Two predictable alternatives avoid this, and a third mitigates it. The first is the
plug-in bet $\lambda_t=\min\{\bar\lambda,\max\{0,\widehat{\E}_{t-1}Y/\widehat{\E}_{t-1}Y^2\}\}$. Here
$Y_s=(Z_s-b_s)/(1+b_s)$, $\widehat{\E}_{t-1}$ is the empirical average over items $1,\dots,t-1$,
and $\bar\lambda<1$ is fixed. This bet approximates GRAPA, the ``growth rate adaptive to the particular
alternative'' bet \citep{waudbysmith2024}, which maximizes the average log-growth over the
past items. Clipping to $[0,\bar\lambda]\subset[0,1]$ meets the requirement of Lemma~\ref{lem:betting}, and the cap
$\bar\lambda<1$ keeps every factor at least $1-\bar\lambda>0$. The bet sets $\lambda_t=0$
whenever $\widehat{\E}_{t-1}Y^2=0$ (in particular $\lambda_1=0$). The second is online Newton
steps \citep{cutkosky2018,waudbysmith2024}, a second-order gradient update of $\lambda_t$ on the running log-wealth, projected onto
$[0,\bar\lambda]$. The third is a \emph{mixture}. For constants $\lambda^{(1)},\dots,\lambda^{(|\Lambda|)}$ let
$\Ep^{(k)}_t=\prod_{s\le t}(1+\lambda^{(k)}Y_s)$ be the wealth of the $k$th constant; the
mixture is the average $|\Lambda|^{-1}\sum_k\Ep^{(k)}_t$. It is an average of
supermartingales for one filtration, hence a supermartingale. It is also the wealth of the
single predictable bet $\lambda_t=\sum_k\Ep^{(k)}_{t-1}\lambda^{(k)}/\sum_k\Ep^{(k)}_{t-1}$, the
wealth-weighted mean of the grid (its denominator is positive because every $\lambda^{(k)}<1$),
so Lemma~\ref{lem:betting} applies verbatim. A mixture over a fixed grid covers the range of
$\Delta/\rho$ the grid spans but not arbitrarily small ratios. Once $\Delta/\rho$ falls below
about half the smallest grid value, every component has negative expected log-growth and the
mixture wealth also tends to zero (Proposition~\ref{prop:power}). The mixture trails the best
constant on the grid by at most a factor $|\Lambda|$ in wealth, that is, by $\log |\Lambda|$ in the threshold,
and has no learning phase. The bet for pair $(j,l)$ may depend on the whole past of all models.
\end{remark}

Under (S) the log-wealth of a constant bet is a sum of independent bounded terms, which yields
the growth rate of every bet on the grid, a consistency condition for the mixture and explicit
certification-time bounds, with no restriction on the dependence between models within an item.
Fix $(j,l)\in\pairs$ with $\Delta=\theta_j-\theta_l>0$, write $Z_t=Z^{jl}_t$, $\Ep_t=\Ep^{jl}_t$ for the grid average of \eqref{eq:update} and
$\Ep^\lambda_t=\Ep^{jl}_{\lambda,t}$ for the wealth of the constant bet $\lambda$, and let
$\lambda_{\min}$ be the smallest element of the grid
$\Lambda=\{\lambda^{(1)},\dots,\lambda^{(|\Lambda|)}\}\subset(0,\tfrac12]$. Let $\mu(\lambda)=\E\log(1+\lambda Z_1)$, $\rho=\E[Z_1^2]$, which is the
discordance rate of Remark~\ref{rem:bets} for binary scores,
$c_\lambda=\log\frac{1+\lambda}{1-\lambda}$,
$v_\lambda=\operatorname{Var}\log(1+\lambda Z_1)$ and $\eta=\log(|\Lambda|\M(\M-1)/\alpha)$. Let
$\tau^{B}=\inf\{t:\Ep_t\ge\M(\M-1)/\alpha\}$ be the e-Bonferroni certification time of $(j,l)$,
and let $\tau^{\mathrm{sc}}$ and $\tau$ be the times at which $(j,l)$ enters $\D^{\mathrm{sc}}_t$ and $\D_t$, respectively, where $\D^{\mathrm{sc}}_t$ is the set
certified by the shortcut (Proposition~\ref{prop:shortcut}). On every path the exact test certifies no later than the shortcut, and the shortcut no
later than e-Bonferroni.

\begin{proposition}[Growth, consistency and certification time under (S)]\label{prop:power}
Assume \textup{(S)}, no retirement, and $\Delta>0$.
\begin{enumerate}[label=(\alph*),leftmargin=*,itemsep=2pt]
\item \emph{Growth.} For $\lambda\in(0,1)$, $t^{-1}\log\Ep^{\lambda}_t\to\mu(\lambda)$ almost
  surely; $\mu$ is concave with $\mu(0)=0$, and for $\lambda\le\tfrac12$,
  $\lambda\Delta-\lambda^2\rho\le\mu(\lambda)\le\lambda\Delta$, so $\mu(\lambda)>0$ whenever
  $0<\lambda<\Delta/\rho$.
\item \emph{Consistency.} If $\mu(\lambda_{\min})>0$ then $\tau^{B}<\infty$ almost surely and
  $(j,l)\in\D_t$ for all $t\ge\tau^{B}$; if $\mu(\lambda_{\min})<0$ then every bet on the grid
  has negative growth and $\Ep_t\to0$ almost surely. The first case holds whenever
  $\lambda_{\min}<\Delta/\rho$; for binary scores,
  $\mu(\lambda)=\tfrac12\rho\,c_\lambda\bigl(\Delta/\rho-x_0(\lambda)\bigr)$ with
  $x_0(\lambda)=\log\frac{1}{1-\lambda^2}/c_\lambda$, so it holds if and only if
  $\Delta/\rho>x_0(\lambda_{\min})$, which is $\Delta/\rho>0.015$ for the default grid. For
  binary scores no finite grid is consistent for every $\Delta>0$: pairs with
  $0<\Delta/\rho<x_0(\lambda_{\min})$ have negative growth on every grid bet. Consistency for
  every positive gap would require bets approaching zero, which we do not pursue here.
\item \emph{Certification time.} Let $\lambda\in\Lambda$ have $\mu=\mu(\lambda)>0$ and write
  $c=c_\lambda$. For every integer $t\ge2\eta/\mu$,
  \begin{equation}\label{eq:tail}
    \Pr(\tau^{B}>t)\le\exp\Bigl(-\frac{t\mu^2}{2c^2}\Bigr)
    \quad\text{and}\quad
    \Pr(\tau^{B}>t)\le\exp\Bigl(-\frac{t\mu^2}{8v_\lambda+\frac43c\mu}\Bigr),
    \qquad v_\lambda\le4\lambda^2\rho .
  \end{equation}
  Define the rate and the horizon
  \[
    a_\lambda=\max\Bigl\{\frac{\mu^2}{2c^2},\ \frac{\mu^2}{8v_\lambda+\frac43c\mu}\Bigr\},
    \qquad
    t^\star(\varepsilon)=\max\Bigl\{\frac{2\eta}{\mu},\ \frac{\log(1/\varepsilon)}{a_\lambda}\Bigr\}.
  \]
  Then, for every
  $\varepsilon\in(0,1)$, $\Pr(\tau^{B}\le\lceil t^\star(\varepsilon)\rceil)\ge1-\varepsilon$,
  and $\E[\tau^{B}]\le2\eta/\mu+1/a_\lambda+2$; the same bounds hold for $\tau^{\mathrm{sc}}$
  and $\tau$.
\item \emph{Default grid.} If $0.06\le\Delta/\rho\le1$, some grid bet lies in
  $[\Delta/(4.2\rho),\Delta/(2\rho)]$, for which $0.18\,\Delta^2/\rho\le\mu\le\Delta^2/(2\rho)$
  and $c\le1.1\,\Delta/\rho$, and then
  \begin{equation}\label{eq:rate}
    t^\star(\varepsilon)\le\frac{11.2\,\rho\eta+\min\{74.7,\,270\rho\}\log(1/\varepsilon)}{\Delta^2},
    \quad
    \E[\tau^{B}]\le\frac{11.2\,\rho\eta+\min\{74.7,\,270\rho\}}{\Delta^2}+2 ,
  \end{equation}
  so certification takes $O(\rho\,\Delta^{-2}\log(|\Lambda|\M^2/(\alpha\varepsilon)))$ items. For
  binary scores with $x_0(0.03)<\Delta/\rho<0.06$ only the general bound above applies, with the exact growth
  $\mu(0.03)=\tfrac12\rho\,c_{0.03}(\Delta/\rho-x_0(0.03))$, which tends to zero as
  $\Delta/\rho\downarrow x_0(0.03)\approx0.015$: a fixed grid has no uniform rate over its
  positive-growth pairs.
\end{enumerate}
\end{proposition}
\subsection{Supplement to the validity theorem}
\label{app:closedsupp}

\begin{remark}[The restriction must be stated through $\T(\wo)$]\label{rem:vacuous}
Every subset $\mathcal{S}\subseteq\pairs$ of the one-sided hypotheses $\Hyp_{jl}$ is
\emph{jointly satisfiable}, since $\theta_1=\dots=\theta_\M$ makes all of them true. Requiring
only that an intersection be consistent with some weak order therefore excludes nothing. The
restriction that does exclude intersections is that of \citet{shaffer1986}: $\mathcal{S}$ needs a test only if it can be the \emph{exact}
set of true hypotheses, that is, $\mathcal{S}=\T(\wo)$ for some $\wo$. The full closure tests the
$2^{\M(\M-1)}-1$ intersection hypotheses $\bigcap_{(j,l)\in\mathcal{S}}\Hyp_{jl}$, one per
nonempty $\mathcal{S}\subseteq\pairs$. Of these, $|\W|$ remain: $13$ of $63$ for $\M=3$, $75$ of
$4095$ for $\M=4$. For $\M=2$ nothing is saved.
\end{remark}

The family of intersections $\bigcap_{(j,l)\in\T(\wo)}\Hyp_{jl}$, $\wo\in\W$, that
Algorithm~\ref{alg:exact} tests is not itself closed under intersection. This does no harm: the
proof of Theorem~\ref{thm:validity} uses only that the cell
$\Theta_{\wostar}=\{\theta:\wo(\theta)=\wostar\}$ containing the truth is rejected with
probability at most $\alpha$. The restriction to exactly realizable
intersections is also that of \citet{westfall1997}.

\begin{proposition}[The output is a partial order]\label{prop:structure}
On the event $\{\CS_t\ne\emptyset\}$, which contains the event in \eqref{eq:cover}, $\D_t$ is a
strict partial order on $\Mset$: irreflexive, asymmetric and transitive. If $\CS_t=\emptyset$ then
$\D_t=\pairs$; this has probability at most $\alpha$ uniformly in $t$, and it reveals that an error
has occurred.
\end{proposition}

\subsection{Testing only orderings loses nothing, and the shortcut is valid}
\label{sec:power}

Testing only the orderings instead of every set of pairwise hypotheses restricts the closed
test. Both claims of Section~\ref{sec:closed} about what this restriction costs hold on every
data set, not merely with high probability. Both compare procedures applied to the same wealths.

\begin{proposition}[Restriction never hurts]\label{prop:dominance}
Let $\D^{\mathrm{full}}_t$ be the output of the unrestricted closed test that rejects $\Hyp_{jl}$
once $\max_{s\le t}|\mathcal{S}|^{-1}\sum_{(a,b)\in\mathcal{S}}\Ep^{ab}_s\ge1/\alpha$ for every
set $\mathcal{S}\subseteq\pairs$ that contains $(j,l)$. Then
$\D^{\mathrm{full}}_t\subseteq\D_t$ for all $t$ on every sample path. For some wealth
configurations, exhibited in Example~\ref{ex:three}, the restricted test certifies a pair that the
full closure does not.
\end{proposition}

\begin{proposition}[Transitivity pooling; validity of Algorithm~\ref{alg:shortcut}]
\label{prop:shortcut}
Let
\[
  B^{jl}_t=\Ep^{jl}_t+\sum_{m\in\Mset\setminus\{j,l\}}\min\{\Ep^{jm}_t,\Ep^{ml}_t\}
\]
be the pooled statistic of Section~\ref{sec:method} and let
$\D^{\mathrm{sc}}_t=\{(j,l)\in\pairs:\max_{s\le t}B^{jl}_s\ge\M(\M-1)/\alpha\}$ be the set that
Algorithm~\ref{alg:shortcut} certifies before closure. Then $\D^{\mathrm{sc}}_t\subseteq\D_t$
for all $t$ on every sample path, so $\D^{\mathrm{sc}}$ inherits \eqref{eq:fwer}. It costs
$O(\M)$ per pair per update and contains the e-Bonferroni rejections
$\{\max_{s\le t}\Ep^{jl}_s\ge\M(\M-1)/\alpha\}$.
\end{proposition}

\paragraph{Relation to existing results.}
To our knowledge both propositions are new. Proposition~\ref{prop:dominance} has a known relative. It is the
sequential counterpart of the fixed-sample observation of \citet{shaffer1986} that logical
relations among hypotheses permit a milder multiplicity correction at no cost in error. It
differs in that it holds for e-processes, on every sample path and at all times, and in that its
comparator $\D^{\mathrm{full}}$ is the e-value closed test of \citet{hartoglei2025}.
Proposition~\ref{prop:shortcut} has, to our knowledge, no direct counterpart. Its closest
relative is the certification logic of \citet{gu2026serpant}, which corrects by Bonferroni over the
pairs (e-Bonferroni in our terms) and uses transitivity only after certification, to propagate certified
pairs. The pooled statistic $B^{jl}_t$ uses transitivity
before certification (Appendix~\ref{app:method}). Because $B^{jl}_t\ge\Ep^{jl}_t$, the shortcut
contains their rejections, at a threshold that stays conservative relative to the exact test.

\begin{lemma}[Transitive closure adds no error]\label{lem:closure}
Let $(\D_t)$ be a certified set with good event $\good$, and let $\Dbar_t$ be the transitive
closure of $\D_t$, the set of $(j,l)\in\pairs$ joined by a directed path $j=v_0,\dots,v_r=l$,
$r\ge1$, with every $(v_{i-1},v_i)\in\D_t$. On $\good$, for every $t$, every
$(j,l)\in\Dbar_t$ satisfies $\thetastar_j>\thetastar_l$; in particular $\Dbar_t$ is a strict
partial order and $\D_t$ contains no directed cycle. Hence $(\Dbar_t)$ is a certified set with
the \emph{same} good event: closing under transitivity adds rejections but never turns an
error-free sample path into an erroneous one, so it leaves the family-wise error guarantee
unchanged.
\end{lemma}

\begin{corollary}[Certified top-$k$ sets]\label{cor:topk}
Let $(\D_t)$ be a certified set with $L_{j,t},U_{j,t}$ as in \eqref{eq:LU}. For $k\in\Mset$ let
$I_{k,t}=\{j:U_{j,t}\le k\}$ and $O_{k,t}=\{j:L_{j,t}>k\}$. On $\good$, for
all $k$ and $t$ simultaneously, every $j\in I_{k,t}$ has $\rank_j(\thetastar)\le k$ and every
$j\in O_{k,t}$ has $\rank_j(\thetastar)>k$.
\end{corollary}

\begin{example}[Strictness, $\M=3$]\label{ex:three}
Let $\theta_1>\theta_2>\theta_3$ and suppose that at some time $t$ the wealths are
$\Ep^{13}\ge6/\alpha$, $\Ep^{12}=a$ and $\Ep^{23}=c$, with the three wealths of true hypotheses
all zero and with no earlier time at which $\Ep^{12}$ or the mean over
$\pairs\setminus\{(1,3)\}$ was larger. Of the eight weak orders with $1\wle_\wo2$, six contain
$(1,3)$ and are eliminated through $\Ep^{13}$; the other two place model $3$ strictly below model
$1$ and hence, by transitivity, below model $2$, so their means are $a/3$ and $a/4$. The
restricted test therefore certifies $1\succ2$ once $a\ge4/\alpha$, whereas the full closure must
also reject $\pairs\setminus\{(1,3)\}$, needing $a+c\ge5/\alpha$, and e-Bonferroni needs
$a\ge6/\alpha$. Whenever $c<1/\alpha$, the restricted test is strictly stronger than both, for
$4/\alpha\le a<5/\alpha-c$: with $\alpha=0.1$, $\Ep^{13}=60$, $a=42$ and
$c=5$, it certifies $1\succ2$ while the full-closure statistic stands at $9.4$ against a
threshold of $10$ and $\Ep^{12}$ at $42$ against $60$.
\end{example}

The bound of Proposition~\ref{prop:shortcut} is crude in two places that could be tightened. It
divides by $\M(\M-1)$ for every ordering, although $|\T(\wo)|=\binom{\M}{2}+\#\{\text{tied pairs}\}$
reaches that value only when every pair is tied. It also has no step-down, that is, no lowering of
the threshold for the remaining orders once some orders are rejected. The exact test has neither
weakness, so what either change could recover is at most what the exact test gains over
e-Bonferroni. In simulation the exact test certifies about $7\%$ more true dominances than
e-Bonferroni after $500$ items and at most $3\%$ more by $4000$. On the eight highest-scoring
models of the real leaderboard it certifies $2$--$5\%$ more up to three quarters of a benchmark
and nearly the same number at the end (Appendix~\ref{app:ecorr}).

\begin{remark}[Weights must not depend on time]\label{rem:weights}
The plain average in \eqref{eq:EW} may be replaced by any weighted form
$\Ep^\wo_t=\sum_{(j,l)\in\T(\wo)}w^\wo_{jl}\,\Ep^{jl}_t$ with fixed weights $w^\wo_{jl}\ge0$
summing to one over $\T(\wo)$, which the proof of Theorem~\ref{thm:validity} covers without change; the weights
may depend on $\wo$. What is \emph{not} valid is letting them adapt: a predictably re-weighted
average $\sum_kw_{k,t}\Ep^{(k)}_t$ of supermartingales is not a supermartingale in general, so
the weights in \eqref{eq:EW} must not depend on $t$. Propositions~\ref{prop:dominance}
and~\ref{prop:shortcut} are stated for the plain average.
\end{remark}

\subsection{Tiers and retirement}
\label{app:tiers}

An \emph{antichain} of a strict partial order $\D$ is a set of models no two of which are
related by $\D$. The tier guarantee of Section~\ref{sec:duality} is part (c) below, with its
monotonicity in (d); the rest is standard order theory.

\begin{theorem}[Tier guarantee, full statement]\label{thm:tiers}
Let $\D$ be a strict partial order on $\Mset$ and compute tiers from $\D$ as in Step~3 of
Section~\ref{sec:method}. Parts (a), (b) and (e) hold for every such $\D$. In (c) and (d),
$\tier_t(j)$ is the tier computed from the transitive closure $\Dbar_t$ of a certified set on
the event $\good$, where $\Dbar_t$ is a strict partial order by Lemma~\ref{lem:closure}, and (d)
assumes in addition $\D_t\subseteq\D_{t+1}$; the subscript is dropped when $t$ is fixed. On $\good$
these tiers equal those computed from $\D_t$ itself: $\D_t\subseteq\Dbar_t$, and each pair of
$\Dbar_t$ is joined by a chain in $\D_t$, which has no cycle on $\good$, so longest chains have the
same length.
\begin{enumerate}[label=(\alph*),leftmargin=*,itemsep=2pt]
\item \emph{Partition into antichains.} The tiers partition $\Mset$, tiers
  $1,\dots,\max_j\tier(j)$ are all nonempty, and no certified dominance holds between two models
  of the same tier.
\item \emph{Order.} If $(j,l)\in\D$ then $\tier(j)<\tier(l)$. Every model in tier $s\ge2$ is
  dominated by some model in tier $s-1$.
\item \emph{Guarantee.} On $\good$, simultaneously for all $j$ and all $t$, the bound
  \eqref{eq:tierbound} holds, that is, $\rank_j(\thetastar)\ge L_{j,t}\ge\tier_t(j)$,
  where $L_{j,t}$ is computed from $\Dbar_t$; the number of tiers is at most the number of
  distinct values among $\thetastar_1,\dots,\thetastar_\M$.
\item \emph{Monotonicity.} If $\D_t\subseteq\D_{t+1}$ then $\tier_t(j)\le\tier_{t+1}(j)$: a model
  is never promoted as evidence accrues.
\item \emph{Minimality.} No partition of $\Mset$ into antichains of $\D$ has fewer blocks
  \citep{mirsky1971}.
\end{enumerate}
\end{theorem}

\begin{remark}[What tiers do not say]\label{rem:tiercaveat}
With $\D=\{(a,b)\}$ on $\{a,b,c\}$ the tiers are $\{a,c\}$ and $\{b\}$, and $c$ is comparable
to neither, which is why $\tier(j)<\tier(l)$ does not imply that $j$ is certified better than
$l$. The intuitive definition ``same tier if and only if neither dominates the other'' fails because
incomparability is not transitive: with $\D=\{(a,c)\}$, $a$ and $b$ are incomparable, $b$ and
$c$ are incomparable, and $a\succ c$.
\end{remark}

The cost of retirement is bounded under (S) for the rule that retires a model once every pair
involving it is resolved. The top-$k$ rule, under which a pair can be frozen before it is
certified, is measured in Appendix~\ref{app:ecost} and not bounded here.

\begin{corollary}[Cost of all-pairs retirement under (S)]\label{cor:cost}
Assume \textup{(S)}, bets that are functions of each pair's own past (the mixture is), and the
\emph{all-pairs rule}, which retires model $j$ at the first time every pair $\{j,l\}$ is
certified in one direction. For an unordered pair let $q(j,l)$ be the ordered pair in the true
direction, and let $\pairs^+=\{(j,l):\theta_j>\theta_l\}$. For $q=(j,l)\in\pairs^+$ write
$\Delta_q=\theta_j-\theta_l$, $\rho_q=\E[(Z^{jl}_1)^2]$ and
$\mu_q(\lambda)=\E\log(1+\lambda Z^{jl}_1)$. These are the quantities $\Delta$, $\rho$ and
$\mu(\lambda)$ of Proposition~\ref{prop:power} computed for the pair $q$. Assume
that every $q\in\pairs^+$ involving $j$ has a bet $\lambda_q\in\Lambda$ with
$\mu_q=\mu_q(\lambda_q)>0$, which holds if and only if $\mu_q(\lambda_{\min})>0$, and let
$a_q=a_{\lambda_q}$ be its rate in Proposition~\ref{prop:power}(c) and $\tau^{B}_q$ the
e-Bonferroni time of $q$ computed without retirement. If model $j$ is tied with no other model, put
$\kappa_j=\max_{l\ne j}2\eta/\mu_{q(j,l)}$ and $\beta_j=\max_{l\ne j}1/a_{q(j,l)}$. Then
$\sigma_j\le\max_{l\ne j}\tau^{B}_{q(j,l)}$ on every path,
$\Pr(\sigma_j>t)\le(\M-1)e^{-t/\beta_j}$ for every integer $t\ge\kappa_j$, and
\begin{equation}\label{eq:cost}
  \E[\sigma_j]\le\kappa_j+\beta_j\bigl(\log(\M-1)+1\bigr)+2 .
\end{equation}
If every model is untied and every true pair satisfies $0.06\le\Delta_q/\rho_q\le1$, then with
$\Delta_j=\min_{l\ne j}|\theta_j-\theta_l|$ and $\rho_j=\max_{l\ne j}\rho_{q(j,l)}$,
\begin{equation}\label{eq:costgrid}
  \E[\sigma_j]\le\frac{11.2\,\rho_j\eta+\min\{74.7,\,270\rho_j\}\bigl(\log(\M-1)+1\bigr)}{\Delta_j^2}+2,
\end{equation}
and consequently $\E[\sum_j\sigma_j]=O\bigl(\sum_j\rho_j\Delta_j^{-2}\log(|\Lambda|\M^2/\alpha)\bigr)$.
A model tied with another model is never retired under the all-pairs rule on the good event
$\good$; over any finite evaluation horizon it therefore incurs the full planned cost.
\end{corollary}

\section{Proofs}
\label{app:proofs}

\newcommand{\why}[1]{\tag*{{\small[#1]}}}

Time is discrete, $t\in\N_0=\{0,1,2,\dots\}$, and $\ind{A}$ denotes the indicator of an event
$A$. We write $A^c$ for the complement of an event $A$ and abbreviate ``almost surely'' to a.s.
A random variable is \emph{$\F_t$-measurable} if its value is determined by the information
$\F_t$. Thus adapted means $\F_t$-measurable at each $t$, and predictable means
$\F_{t-1}$-measurable. Conditional expectations of nonnegative random variables are always
defined, possibly infinite, and the tower property holds for them; we use this without comment.
In each proof, the bracket at the right of a line states what justifies that line.

\subsection{Pairwise wealths}

\begin{proof}[Proof of Lemma~\ref{lem:betting}]
Write $\phi_t=1+\lambda_t(Z_t-b_t)/(1+b_t)$, so that $\Ep^{jl}_t=\prod_{s\le t}\phi_s$ by the
definition of the wealth in \eqref{eq:wealth}.

\emph{Nonnegativity and measurability.}
\begin{align*}
  \frac{Z_t-b_t}{1+b_t} &\ \ge\ \frac{-1-b_t}{1+b_t}=-1
    \why{$Z_t\ge-1$ and $1+b_t>0$}\\
  \phi_t &\ \ge\ 1-\lambda_t\ \ge\ 0
    \why{previous line and $\lambda_t\in[0,1]$}\\
  \Ep^{jl}_t=\prod_{s\le t}\phi_s &\ \ge\ 0,\ \text{and is }\F_t\text{-measurable}
    \why{$Z_s$, $b_s$, $\lambda_s$ are $\F_s$-measurable for $s\le t$}
\end{align*}

\emph{Bound on one factor.} To avoid any integrability question for the signed quantity
$\lambda_t(Z_t-b_t)/(1+b_t)$, rewrite $\phi_t$ as a sum of two nonnegative terms:
\begin{align*}
  \phi_t &= 1-\lambda_t+\frac{\lambda_t}{1+b_t}\,(1+Z_t)
    \why{algebra; both terms are nonnegative since $\lambda_t\le1$, $1+Z_t\ge0$, $1+b_t>0$}\\
  \E_P[\phi_t\mid\F_{t-1}] &= 1-\lambda_t+\frac{\lambda_t}{1+b_t}\bigl(1+\E_P[Z_t\mid\F_{t-1}]\bigr)
    \why{$\lambda_t$ and $b_t$ are $\F_{t-1}$-measurable}\\
  &\ \le\ 1-\lambda_t+\frac{\lambda_t}{1+b_t}(1+b_t)=1
    \why{$\E_P[Z_t\mid\F_{t-1}]\le b_t$ and $\lambda_t/(1+b_t)\ge0$}
\end{align*}

\emph{Supermartingale.}
\begin{align*}
  \E_P[\Ep^{jl}_t\mid\F_{t-1}] &= \Ep^{jl}_{t-1}\,\E_P[\phi_t\mid\F_{t-1}]
    \why{$\Ep^{jl}_{t-1}\ge0$ is $\F_{t-1}$-measurable}\\
  &\ \le\ \Ep^{jl}_{t-1}
    \why{bound on one factor}\\
  \E_P[\Ep^{jl}_t] &\ \le\ \E_P[\Ep^{jl}_{t-1}]\ \le\ \dots\ \le\ \Ep^{jl}_0=1
    \why{tower property; induction on $t$}
\end{align*}
so the process is integrable and is a nonnegative supermartingale started at $1$.

\emph{Frozen process.} Let $\sigma$ be an $(\F_t)$-stopping time.
\begin{align*}
  \Ep^{jl}_{t\wedge\sigma}-\Ep^{jl}_{(t-1)\wedge\sigma}
    &= \ind{\sigma\ge t}\bigl(\Ep^{jl}_t-\Ep^{jl}_{t-1}\bigr)
    \why{both sides vanish on $\{\sigma<t\}$ and agree on $\{\sigma\ge t\}$}\\
  \E_P\bigl[\Ep^{jl}_{t\wedge\sigma}-\Ep^{jl}_{(t-1)\wedge\sigma}\,\big|\,\F_{t-1}\bigr]
    &= \ind{\sigma\ge t}\,\bigl(\E_P[\Ep^{jl}_t\mid\F_{t-1}]-\Ep^{jl}_{t-1}\bigr)\ \le\ 0
    \why{$\{\sigma\ge t\}=\{\sigma\le t-1\}^c\in\F_{t-1}$; $\Ep^{jl}_t$, $\Ep^{jl}_{t-1}$ integrable by the previous block; supermartingale property}
\end{align*}
so the stopped process is again a nonnegative supermartingale started at $1$.
\end{proof}

\begin{proof}[Proof of Proposition~\ref{prop:instances}]
By Lemma~\ref{lem:betting} it suffices to verify, whenever $\theta_j\le\theta_l$, that $b_t$ is
predictable with $b_t\in(-1,1]$ and that $\E[Z_t\mid\F_{t-1}]\le b_t$.

(a) Under (F), let $z_i=x_{ij}-x_{il}$, so that $\Hyp_{jl}$ reads $\sum_{i=1}^Nz_i\le0$. Fix
$t\le N$.
\begin{align*}
  \E[Z_t\mid\F_{t-1}] &= \frac{\sum_{i=1}^Nz_i-S_{t-1}}{N-t+1}
    \why{given $\F_{t-1}$, $\Pi(t)$ is uniform on the $N-t+1$ unrevealed items, and $\xi$ is independent of $\Pi$}\\
  &\ \le\ \frac{-S_{t-1}}{N-t+1}
    \why{$\sum_{i=1}^Nz_i\le0$}\\
  \E[Z_t\mid\F_{t-1}] &\ \le\ 1
    \why{$Z_t\le1$, scores lying in $[0,1]$}\\
  \E[Z_t\mid\F_{t-1}] &\ \le\ \min\Bigl\{1,\frac{-S_{t-1}}{N-t+1}\Bigr\}\ \le\ b_t
    \why{the two lines above; taking the maximum with $-1+\delta$ in eq.~\eqref{eq:bfinite} can only raise the bound}\\
  -1<-1+\delta\ \le\ b_t &\ \le\ 1
    \why{eq.~\eqref{eq:bfinite}: the clipping, as Lemma~\ref{lem:betting} requires; the minimum with $1$ and $\delta\le1$}\\
  b_t &\ \text{is }\F_{t-1}\text{-measurable}
    \why{it is a function of $S_{t-1}=\sum_{s<t}Z_s$}
\end{align*}
For $t>N$ the process is constant. All $\M(\M-1)$ processes are adapted to the single filtration
generated by the shared permutation; independence of permutations across pairs is neither
available nor needed.

(b) Under (S), with $b_t\equiv0$:
\begin{align*}
  \E[Z_t\mid\F_{t-1}] &= \E[X_{tj}-X_{tl}]
    \why{$X_t$ is independent of $\F_{t-1}=\sigma(X_1,\dots,X_{t-1},\xi)$}\\
  &= \theta_j-\theta_l\ \le\ 0=b_t
    \why{definition of $\theta$; $\Hyp_{jl}$}
\end{align*}
Nothing is assumed about the joint law of $(X_{t1},\dots,X_{t\M})$, and $\lambda_t$ may depend on
the past of all models, since only $\F_{t-1}$-measurability is used.
\end{proof}

\begin{proof}[Proof of Proposition~\ref{prop:power}]
Under (S) with $b_t\equiv0$ the factor of the constant bet $\lambda$ on item $s$ is
$1+\lambda Z_s$; write $\ell_s=\log(1+\lambda Z_s)$, so that
$\log\Ep^{\lambda}_t=\sum_{s\le t}\ell_s$.

\emph{(a) Growth.} Let $f(u)=\log(1+u)-u+u^2$; then $f(0)=0$ and $f'(u)=u(1+2u)/(1+u)$, which
is $\le0$ on $[-\tfrac12,0]$ and $\ge0$ on $[0,\tfrac12]$, so $f\ge0$ on $[-\tfrac12,\tfrac12]$.
\begin{align*}
  \ell_1,\ell_2,\dots\ \text{are i.i.d.\ with}\ |\ell_s| &\le \log\tfrac{1}{1-\lambda}
    \why{$Z_s$ i.i.d., $|Z_s|\le1$, $\lambda<1$}\\
  t^{-1}\textstyle\sum_{s\le t}\ell_s &\to \E\ell_1=\mu(\lambda)\ \text{a.s.}
    \why{strong law of large numbers}\\
  \mu\ \text{is concave on }[0,1)\text{, and }\mu(0) &= 0
    \why{$\lambda\mapsto\log(1+\lambda z)$ is concave for each $z\in[-1,1]$; expectation preserves concavity}\\
  u-u^2\le\log(1+u) &\le u\quad\text{for }u\in[-\tfrac12,\tfrac12]
    \why{$f\ge0$; concavity of $\log$}\\
  \lambda\Delta-\lambda^2\rho\le\mu(\lambda) &\le \lambda\Delta\quad\text{for }\lambda\le\tfrac12
    \why{$u=\lambda Z_1\in[-\tfrac12,\tfrac12]$; take expectations, $\E Z_1=\Delta$, $\E Z_1^2=\rho$}
\end{align*}

\emph{(b) Consistency.} Suppose $\mu(\lambda_{\min})>0$.
\begin{align*}
  \log\Ep_t &\ge \log\Ep^{\lambda_{\min}}_t-\log |\Lambda|\to+\infty\ \text{a.s.}
    \why{the mixture averages nonnegative terms; (a)}\\
  \tau^{B} &< \infty\ \text{a.s.}
    \why{$\tau^{B}$ is the first time $\Ep_t\ge\M(\M-1)/\alpha$}\\
  (j,l) &\in \D_t\ \text{for all }t\ge\tau^{B}
    \why{e-Bonferroni uses the running maximum and is contained in $\D^{\mathrm{sc}}_t\subseteq\D_t$ (Proposition~\ref{prop:shortcut}); $\D_t$ is monotone}
\end{align*}
Suppose $\mu(\lambda_{\min})<0$ and let $\lambda\in\Lambda$ with $\lambda>\lambda_{\min}$.
\begin{align*}
  \mu(\lambda_{\min}) &\ge \tfrac{\lambda_{\min}}{\lambda}\mu(\lambda)+\bigl(1-\tfrac{\lambda_{\min}}{\lambda}\bigr)\mu(0)
    \why{concavity of $\mu$}\\
  \mu(\lambda) &\le \tfrac{\lambda}{\lambda_{\min}}\mu(\lambda_{\min})<0
    \why{$\mu(0)=0$; rearrange}\\
  \Ep_t &\to 0\ \text{a.s.}
    \why{(a): each of the $|\Lambda|$ components has $\log\Ep^{\lambda}_t\to-\infty$}
\end{align*}
For binary scores, with $p_\pm=\Pr(Z_1=\pm1)=(\rho\pm\Delta)/2$ and $\rho>0$:
\begin{align*}
  \mu(\lambda)=p_+\log(1+\lambda)+p_-\log(1-\lambda) &= \tfrac12\rho\log(1-\lambda^2)+\tfrac12\Delta\,c_\lambda
    \why{$Z_1\in\{-1,0,1\}$; $\log(1+\lambda)\pm\log(1-\lambda)$ equal $\log(1-\lambda^2)$ and $c_\lambda$}\\
  &= \tfrac12\rho\,c_\lambda\bigl(\Delta/\rho-x_0(\lambda)\bigr),\qquad x_0(\lambda)=\log\tfrac{1}{1-\lambda^2}/c_\lambda
    \why{factor out $\tfrac12\rho c_\lambda>0$}\\
  \mu(\lambda_{\min})>0 &\iff \Delta/\rho>x_0(\lambda_{\min})
    \why{previous line; $x_0(0.03)=0.000900/0.0600=0.0150$}
\end{align*}

\emph{(c) Certification time.} Fix $\lambda\in\Lambda$ with $\mu>0$ and put $V_s=\ell_s-\mu$.
\begin{align*}
  \ell_s &\in [\log(1-\lambda),\log(1+\lambda)],\ \text{an interval of length }c
    \why{$|Z_s|\le1$}\\
  \{\tau^{B}>t\} &\subseteq \{\Ep_t<\tfrac{\M(\M-1)}{\alpha}\}\subseteq\{\Ep^{\lambda}_t<\tfrac{|\Lambda|\M(\M-1)}{\alpha}\}=\{\textstyle\sum_{s\le t}\ell_s<\eta\}
    \why{first crossing; $\Ep_t\ge\Ep^{\lambda}_t/|\Lambda|$; definition of $\eta$}\\
  \Pr\bigl(\textstyle\sum_{s\le t}\ell_s\le t\mu-x\bigr) &\le \exp\bigl(-2x^2/(tc^2)\bigr),\quad x>0
    \why{Hoeffding's inequality, i.i.d.\ summands with range $c$}\\
  \Pr(\tau^{B}>t) &\le \exp\bigl(-t\mu^2/(2c^2)\bigr)\quad\text{for integer }t\ge2\eta/\mu
    \why{the two previous lines with $x=t\mu-\eta\ge t\mu/2$; the bound decreases in $x$}\\
  |\log(1+u)| &\le 2|u|\quad\text{for }u\in[-\tfrac12,\tfrac12]
    \why{$u\ge0$: $\log(1+u)\le u$; $u<0$: $-\log(1+u)\le\tfrac{1}{1+u}-1=\tfrac{-u}{1+u}\le2|u|$}\\
  v_\lambda &\le \E\ell_1^2\le4\lambda^2\rho
    \why{previous line with $u=\lambda Z_1$}\\
  \Pr\bigl(\textstyle\sum_{s\le t}V_s\le-x\bigr) &\le \exp\Bigl(-\tfrac{x^2}{2(tv_\lambda+cx/3)}\Bigr),\quad x>0
    \why{Bernstein's inequality: $V_s$ i.i.d., centered, $|V_s|\le c$}\\
  \Pr(\tau^{B}>t) &\le \exp\Bigl(-\tfrac{(t\mu/2)^2}{2(tv_\lambda+ct\mu/6)}\Bigr)=\exp\Bigl(-\tfrac{t\mu^2}{8v_\lambda+\frac43c\mu}\Bigr)\quad\text{for integer }t\ge2\eta/\mu
    \why{$\{\tau^{B}>t\}\subseteq\{\sum_{s\le t}V_s<-(t\mu-\eta)\}$ with $t\mu-\eta\ge t\mu/2$; the bound decreases in $x$}
\end{align*}
which is \eqref{eq:tail}; since both bounds hold, $\Pr(\tau^{B}>t)\le e^{-a_\lambda t}$ for
every integer $t\ge2\eta/\mu$, with $a_\lambda$ the larger of the two rates. Write $a=a_\lambda$.
\begin{align*}
  \Pr(\tau^{B}>\lceil t^\star(\varepsilon)\rceil) &\le \varepsilon
    \why{$\lceil t^\star\rceil\ge2\eta/\mu$ and $a\lceil t^\star\rceil\ge\log(1/\varepsilon)$}\\
  \E[\tau^{B}]=\textstyle\sum_{t\ge0}\Pr(\tau^{B}>t) &\le t_0+\textstyle\sum_{t\ge t_0}e^{-at}\le t_0+1+\tfrac1a,\quad t_0=\lceil2\eta/\mu\rceil
    \why{probabilities are at most one below $t_0$; geometric series; $1/(1-e^{-a})\le1+1/a$}\\
  \E[\tau^{B}] &\le 2\eta/\mu+1/a+2
    \why{$t_0\le2\eta/\mu+1$}\\
  \text{the bounds hold for }\tau^{\mathrm{sc}} &\text{ and }\tau
    \why{$\tau\le\tau^{\mathrm{sc}}\le\tau^{B}$ on every path, Proposition~\ref{prop:shortcut}}
\end{align*}

\emph{(d) Default grid.} Let $I=[\Delta/(4.2\rho),\Delta/(2\rho)]$. If
$0.06\le\Delta/\rho\le0.126$ then $0.03\in I$, since $0.03\le\Delta/(2\rho)$ iff
$\Delta/\rho\ge0.06$ and $0.03\ge\Delta/(4.2\rho)$ iff $\Delta/\rho\le0.126$. If
$0.126<\Delta/\rho\le1$, sort the grid increasingly and let $\lambda^{(k)}$ be the largest grid
bet at most $\Delta/(4.2\rho)$; either $\lambda^{(k)}=\Delta/(4.2\rho)\in I$, or
$\lambda^{(k)}<\Delta/(4.2\rho)\le1/4.2<0.5$, so that $\lambda^{(k+1)}$ exists and, consecutive grid
ratios being at most $2.084<2.1$,
$\Delta/(4.2\rho)<\lambda^{(k+1)}\le2.084\,\lambda^{(k)}<\Delta/(2\rho)$. In both cases some
grid bet $\lambda$ lies in $I$, and for it:
\begin{align*}
  \mu\ge\lambda(\Delta-\lambda\rho) &\ge 0.18\,\Delta^2/\rho
    \why{(a); $\lambda(\Delta-\lambda\rho)$ is concave in $\lambda$, with values $0.181\,\Delta^2/\rho$ and $0.25\,\Delta^2/\rho$ at the ends of $I$}\\
  \mu\le\lambda\Delta\le\tfrac{\Delta^2}{2\rho},\qquad c &\le 2.2\lambda\le1.1\,\tfrac{\Delta}{\rho}
    \why{(a); $c_\lambda/\lambda=2\operatorname{artanh}(\lambda)/\lambda$ increases to $2.197$ at $\lambda=\tfrac12$}\\
  \tfrac{2\eta}{\mu}\le\tfrac{11.2\,\rho\eta}{\Delta^2},\qquad \tfrac{2c^2}{\mu^2} &\le \tfrac{2(1.21)}{0.0324}\,\Delta^{-2}=74.7\,\Delta^{-2}
    \why{substitute the bounds on $\mu$ and $c$}\\
  8v_\lambda\le32\lambda^2\rho\le\tfrac{8\Delta^2}{\rho},\qquad \tfrac43c\mu &\le \tfrac43(1.1)\tfrac12\,\tfrac{\Delta^3}{\rho^2}\le0.734\,\tfrac{\Delta^2}{\rho}
    \why{$\lambda\le\Delta/(2\rho)$; the bounds on $c$ and $\mu$; $\Delta\le\rho$ in this regime}\\
  \tfrac{8v_\lambda+\frac43c\mu}{\mu^2} &\le \tfrac{8.734\,\Delta^2/\rho}{0.0324\,\Delta^4/\rho^2}\le270\,\tfrac{\rho}{\Delta^2}
    \why{the two previous lines}
\end{align*}
and \eqref{eq:rate} follows from (c) with $1/a_\lambda\le\min\{74.7,270\rho\}\Delta^{-2}$. For
binary scores the exact growth of $\lambda=0.03$ is the formula in (b), with $c_{0.03}=0.0600$.
\end{proof}
\subsection{Validity, structure and power}

\begin{proof}[Proof of Theorem~\ref{thm:validity}]
Fix the data-generating law $P$, with parameter $\thetastar$ and $\wostar=\wo(\thetastar)$, and
write $\Ep^{\wo}=(\Ep^{\wo}_t)_{t\ge0}$ for the process of \eqref{eq:EW}.

\emph{$\Ep^{\wostar}$ is an e-process under $P$.} Let $\tau$ be a bounded
$(\F_t)$-stopping time.
\begin{align*}
  (j,l)\in\T(\wostar) &\ \Rightarrow\ \thetastar_j\le\thetastar_l
    \why{definition of $\T(\wo)$}\\
  &\ \Rightarrow\ \E_P\bigl[\Ep^{jl}_\tau\bigr]\le1
    \why{Assumption~\ref{ass:eproc}; $\tau$ is a stopping time for the common filtration}\\
  \E_P\bigl[\Ep^{\wostar}_\tau\bigr]
    &= \sum_{(j,l)\in\T(\wostar)}w^{\wostar}_{jl}\,\E_P\bigl[\Ep^{jl}_\tau\bigr]
    \why{eq.~\eqref{eq:EW} in its weighted form (Remark~\ref{rem:weights}); linearity, the weights being deterministic}\\
  &\ \le\ \sum_{(j,l)\in\T(\wostar)}w^{\wostar}_{jl}=1
    \why{$w^{\wostar}_{jl}\ge0$, the previous line, and the weights sum to one}
\end{align*}
Only linearity of expectation is used; the joint law of the pairwise processes plays no role.
That all of them are e-processes for the \emph{same} filtration is needed, because $\tau$ must
be a stopping time for each of them at once.

\emph{Ville's inequality.} Let $\tau^\star=\inf\{t:\Ep^{\wostar}_t\ge1/\alpha\}$ and fix
a deterministic $n$.
\begin{align*}
  \tau^\star &\ \text{is an }(\F_t)\text{-stopping time}
    \why{$\Ep^{\wostar}$ is adapted}\\
  \Ep^{\wostar}_{\tau^\star\wedge n} &\ \ge\ \alpha^{-1}\ind{\tau^\star\le n}
    \why{nonnegativity on $\{\tau^\star>n\}$; the threshold on $\{\tau^\star\le n\}$}\\
  1 &\ \ge\ \E_P\bigl[\Ep^{\wostar}_{\tau^\star\wedge n}\bigr]
    \why{the e-process property above, at the bounded stopping time $\tau^\star\wedge n$}\\
  &\ \ge\ \alpha^{-1}\Pr(\tau^\star\le n)
    \why{the pathwise bound of the previous display}\\
  \Pr(\tau^\star<\infty) &\ \le\ \alpha
    \why{let $n\to\infty$}\\
  \{\exists t:\wostar\notin\CS_t\} &= \{\tau^\star<\infty\}
    \why{definition of $\CS_t$}
\end{align*}
which is \eqref{eq:cover}.

\emph{Family-wise error.} On $\{\tau^\star=\infty\}$, for every $t$:
\begin{align*}
  \wostar &\in\CS_t
    \why{definition of $\tau^\star$ and definition of $\CS_t$}\\
  (j,l)\in\D_t &\ \Rightarrow\ l\wlt_\wo j\ \text{for every }\wo\in\CS_t
    \why{eq.~\eqref{eq:D}}\\
  &\ \Rightarrow\ l\wlt_{\wostar}j
    \why{$\wostar\in\CS_t$}\\
  &\ \Rightarrow\ \thetastar_j>\thetastar_l
    \why{$\wostar=\wo(\thetastar)$}
\end{align*}
So the event in \eqref{eq:fwer} is contained in $\{\tau^\star<\infty\}$, whose probability is at
most $\alpha$ by Ville's inequality above.

\emph{Random times.} Both events are of the form $\{\forall t:\dots\}$, so they imply the
corresponding statement at any $\N_0$-valued random time, whether or not it is a stopping time.
\end{proof}

\begin{proof}[Proof of Proposition~\ref{prop:structure}]
Suppose $\CS_t\ne\emptyset$ and let $\wo\in\CS_t$ be arbitrary.

\emph{Irreflexive.} $\D_t\subseteq\pairs$, which contains no pair $(j,j)$.

\emph{Asymmetric.}
\begin{align*}
  (j,l),(l,j)\in\D_t &\ \Rightarrow\ l\wlt_\wo j\ \text{and}\ j\wlt_\wo l
    \why{eq.~\eqref{eq:D}}\\
  &\ \Rightarrow\ \text{a contradiction}
    \why{$l\wlt_\wo j$ means $l\wle_\wo j$ and not $j\wle_\wo l$}
\end{align*}

\emph{Transitive.} Let $(j,l),(l,m)\in\D_t$.
\begin{align*}
  m &\ne j
    \why{asymmetry}\\
  l\wlt_\wo j\ &\text{and}\ m\wlt_\wo l
    \why{eq.~\eqref{eq:D}}\\
  m &\wle_\wo j
    \why{$\wle_\wo$ is transitive}\\
  j\wle_\wo m\ &\text{is impossible}
    \why{$j\wle_\wo m\wle_\wo l$ would give $j\wle_\wo l$, contradicting $l\wlt_\wo j$}\\
  m &\wlt_\wo j
    \why{the two previous lines}\\
  (j,m) &\in\D_t
    \why{$\wo\in\CS_t$ was arbitrary; eq.~\eqref{eq:D}}
\end{align*}

\emph{Empty confidence set.} If $\CS_t=\emptyset$:
\begin{align*}
  \D_t &= \pairs
    \why{the condition in eq.~\eqref{eq:D} is vacuous}\\
  \{\exists t:\CS_t=\emptyset\} &\subseteq \{\exists t:\wostar\notin\CS_t\}
    \why{$\wostar\in\CS_t$ would make $\CS_t$ nonempty}\\
  \Pr(\exists t:\CS_t=\emptyset) &\ \le\ \alpha
    \why{eq.~\eqref{eq:cover}}
\end{align*}
\end{proof}

\begin{proof}[Proof of Proposition~\ref{prop:dominance}]
Let $(j,l)\in\D^{\mathrm{full}}_t$ and let $\wo$ be any weak order with $(j,l)\in\T(\wo)$; put
$\mathcal{S}=\T(\wo)$.
\begin{align*}
  \mathcal{S} &\subseteq\pairs\ \text{and}\ \mathcal{S}\ni(j,l)
    \why{definition of $\T(\wo)$ and the choice of $\wo$}\\
  \max_{s\le t}|\mathcal{S}|^{-1}\sum_{(a,b)\in\mathcal{S}}\Ep^{ab}_s &\ \ge\ 1/\alpha
    \why{$(j,l)\in\D^{\mathrm{full}}_t$: the full closure rejects every such $\mathcal{S}$}\\
  \Ep^\wo_s &= |\mathcal{S}|^{-1}\sum_{(a,b)\in\mathcal{S}}\Ep^{ab}_s
    \why{eq.~\eqref{eq:EW} with arithmetic-mean weights}\\
  \max_{s\le t}\Ep^\wo_s &\ \ge\ 1/\alpha,\ \text{so }\wo\notin\CS_t
    \why{the two previous lines; definition of $\CS_t$}\\
  (j,l) &\in\D_t
    \why{$\wo$ with $(j,l)\in\T(\wo)$ was arbitrary; eq.~\eqref{eq:D}}
\end{align*}
Example~\ref{ex:three} gives a configuration in which the inclusion is strict; it has been
checked by enumerating the $13$ weak orders and the $63$ subsets.
\end{proof}

\begin{proof}[Proof of Proposition~\ref{prop:shortcut}]
Let $(j,l)\in\D^{\mathrm{sc}}_t$ and choose $s\le t$ with $B^{jl}_s\ge\M(\M-1)/\alpha$. Let
$\wo$ be any weak order with $j\wle_\wo l$, that is $(j,l)\in\T(\wo)$, and let
$m\in\Mset\setminus\{j,l\}$.

\emph{Claim: $j\wle_\wo m$ or $m\wle_\wo l$.}
\begin{align*}
  \text{neither }j\wle_\wo m\text{ nor }m\wle_\wo l &\ \Rightarrow\ m\wlt_\wo j\ \text{and}\ l\wlt_\wo m
    \why{totality of $\wle_\wo$}\\
  &\ \Rightarrow\ l\wlt_\wo j
    \why{transitivity}\\
  &\ \Rightarrow\ \text{a contradiction with }j\wle_\wo l
    \why{$l\wlt_\wo j$ means not $j\wle_\wo l$}
\end{align*}

\emph{Pooling.}
\begin{align*}
  \T(\wo) &\ni (j,l)\ \text{and, for each }m\in\Mset\setminus\{j,l\},\ (j,m)\text{ or }(m,l)
    \why{definition of $\T(\wo)$ and the claim}\\
  \text{these }\M-1\text{ pairs} &\ \text{are distinct}
    \why{two could coincide only if $(j,m)=(m',l)$ for some $m,m'\notin\{j,l\}$, which would force $m'=j$}\\
  \sum_{(a,b)\in\T(\wo)}\Ep^{ab}_s &\ \ge\ \Ep^{jl}_s+\sum_{m\in\Mset\setminus\{j,l\}}\min\{\Ep^{jm}_s,\Ep^{ml}_s\}=B^{jl}_s
    \why{all e-processes are nonnegative; the pair chosen for $m$ contributes at least the minimum}\\
  \Ep^\wo_s &\ \ge\ \frac{B^{jl}_s}{\M(\M-1)}\ \ge\ \frac1\alpha
    \why{eq.~\eqref{eq:EW} with $|\T(\wo)|\le\M(\M-1)$; the choice of $s$}\\
  \wo &\notin\CS_t
    \why{definition of $\CS_t$, since $s\le t$}\\
  (j,l) &\in\D_t
    \why{$\wo$ with $(j,l)\in\T(\wo)$ was arbitrary; eq.~\eqref{eq:D}}
\end{align*}
The same time $s$ serves for all $\wo$, which is why the running maximum may be taken over
$B^{jl}$. The e-Bonferroni rejections are contained in $\D^{\mathrm{sc}}_t$ because
$B^{jl}_s\ge\Ep^{jl}_s$. The cost is one sum of $\M-2$ minima per pair.
\end{proof}

\subsection{Closure, ranks, tiers and retirement}

\begin{proof}[Proof of Lemma~\ref{lem:closure}]
Work on $\good$ and fix $t$. Let $(j,l)\in\Dbar_t$, joined by the directed path
$v_0=j,v_1,\dots,v_r=l$ with $(v_{i-1},v_i)\in\D_t$ for all $i$.
\begin{align*}
  \thetastar_{v_{i-1}} &> \thetastar_{v_i}\quad\text{for every }i
    \why{$(v_{i-1},v_i)\in\D_t$ and definition of $\good$}\\
  \thetastar_j=\thetastar_{v_0} &> \thetastar_{v_1}>\dots>\thetastar_{v_r}=\thetastar_l
    \why{chain the previous line}\\
  &\D_t\ \text{has no directed cycle}
    \why{a cycle through $j$ is such a path from $j$ to itself and would give $\thetastar_j>\thetastar_j$}\\
  &\Dbar_t\ \text{is a strict partial order}
    \why{transitive by construction; irreflexive because $(j,j)\in\Dbar_t$ would be a directed cycle}\\
  &\text{the good event of }(\Dbar_t)\text{ contains }\good
    \why{the second line holds on $\good$ for every $t$ and every $(j,l)\in\Dbar_t$}\\
  &\text{the good event of }(\Dbar_t)\text{ is contained in }\good
    \why{$\D_t\subseteq\Dbar_t$ for every $t$}
\end{align*}
so the two events coincide.
\end{proof}

\begin{proof}[Proof of Theorem~\ref{thm:ranks}]
(a) On the event in \eqref{eq:cover}, for all $j$ and $t$:
\begin{align*}
  \wostar &\in\CS_t
    \why{eq.~\eqref{eq:cover}}\\
  \rank_j(\thetastar) &= \rank_j(\wostar)
    \why{the rank depends on $\theta$ only through $\wo(\theta)$}\\
  &\in\Rset_{j,t}
    \why{definition of $\Rset_{j,t}$ and the first line}
\end{align*}

(b) Work on $\good$ and fix $j$ and $t$.
\begin{align*}
  (l,j)\in\D_t &\ \Rightarrow\ \thetastar_l>\thetastar_j
    \why{definition of $\good$}\\
  \#\{l:\thetastar_l>\thetastar_j\} &\ \ge\ \#\{l:(l,j)\in\D_t\}
    \why{previous line}\\
  \rank_j(\thetastar) &\ \ge\ 1+\#\{l:(l,j)\in\D_t\}=L_{j,t}
    \why{eq.~\eqref{eq:rank} and eq.~\eqref{eq:LU}}\\
  (j,l)\in\D_t &\ \Rightarrow\ l\ne j\ \text{and}\ \thetastar_l<\thetastar_j
    \why{$\D_t\subseteq\pairs$; definition of $\good$}\\
  \#\{l:\thetastar_l>\thetastar_j\} &\ \le\ (\M-1)-\#\{l:(j,l)\in\D_t\}
    \why{each $l$ with $(j,l)\in\D_t$ is one of the $\M-1$ other models and is not counted on the left}\\
  \rank_j(\thetastar) &\ \le\ \M-\#\{l:(j,l)\in\D_t\}=U_{j,t}
    \why{eq.~\eqref{eq:rank} and eq.~\eqref{eq:LU}}
\end{align*}
Ties cause no difficulty: models tied with $j$ are counted on neither side. Since
$\D_t\subseteq\Dbar_t$ and $\Dbar_t$ is certified with the same good event
(Lemma~\ref{lem:closure}), using $\Dbar_t$ can only raise $L$ and lower $U$.

(c) Let $\wo\in\CS_t$; recall $\rank_j(\wo)=1+\#\{l:j\wlt_\wo l\}$.
\begin{align*}
  (l,j)\in\D_t &\ \Rightarrow\ j\wlt_\wo l,\ \text{so }l\text{ is counted in }\rank_j(\wo)
    \why{eq.~\eqref{eq:D}}\\
  \rank_j(\wo) &\ \ge\ 1+\#\{l:(l,j)\in\D_t\}=L_{j,t}
    \why{previous line}\\
  (j,l)\in\D_t &\ \Rightarrow\ l\wlt_\wo j,\ \text{so }l\text{ is not counted}
    \why{eq.~\eqref{eq:D}; $l\wlt_\wo j$ excludes $j\wlt_\wo l$}\\
  \rank_j(\wo) &\ \le\ \M-\#\{l:(j,l)\in\D_t\}=U_{j,t}
    \why{at most the $\M-1$ other models are counted, and those with $(j,l)\in\D_t$ are not}
\end{align*}
Hence $\Rset_{j,t}\subseteq[L_{j,t},U_{j,t}]$.

\emph{Strictness.} Take $\M=3$, $\alpha=0.1$ and, as hypothetical values at some time with no
earlier crossing, $\Ep^{12}=10$, $\Ep^{13}=90$, $\Ep^{21}=0.2$, $\Ep^{32}=21$,
$\Ep^{23}=\Ep^{31}=0$. Enumeration shows that exactly four weak orders survive, namely
$3\weq2\wlt1$, \ $3\wlt2\weq1$, \ $2\wlt3\wlt1$ and $3\wlt2\wlt1$, writing each from its lowest-ranked model to its highest (so $3\weq2\wlt1$ ranks model 1 first and ties models 2 and 3), and that $\D_t=\{(1,3)\}$.
Hence $U_{1,t}=2$, whereas model $1$ has rank $1$ in every surviving weak order, so
$\Rset_{1,t}=\{1\}$. The gain comes from the tie convention: $\Hyp_{12}$ cannot be rejected
while the tie $1\weq2$ survives, yet a tie at the top still gives rank $1$. Projection sets need
not be intervals: with $(\Ep^{12},\Ep^{13},\Ep^{21},\Ep^{23},\Ep^{31},\Ep^{32})=(0,5,0,31,5,31)$
and $\alpha=0.1$, again with no earlier crossing, exactly $1\weq2\wlt3$ and $3\wlt1\weq2$
survive, so $\Rset_{3,t}=\{1,3\}$ while $[L_{3,t},U_{3,t}]=[1,3]$.
\end{proof}

\begin{proof}[Proof of Corollary~\ref{cor:topk}]
On $\good$, for all $j$, $k$ and $t$ at once:
\begin{align*}
  j\in I_{k,t} &\ \Rightarrow\ \rank_j(\thetastar)\ \le\ U_{j,t}\ \le\ k
    \why{Theorem~\ref{thm:ranks}(b); definition of $I_{k,t}$}\\
  j\in O_{k,t} &\ \Rightarrow\ \rank_j(\thetastar)\ \ge\ L_{j,t}\ >\ k
    \why{Theorem~\ref{thm:ranks}(b); definition of $O_{k,t}$}
\end{align*}
\end{proof}

\begin{proof}[Proof of Theorem~\ref{thm:tiers}]
Let $\D$ be a strict partial order on the finite set $\Mset$. For (c) work on $\good$, fix $t$,
and take $\D=\Dbar_t$. Part (d) is deterministic and is proved for any nested pair whose closures
are strict partial orders.
\begin{align*}
  &\text{a chain cannot repeat an element}
    \why{a repetition would give $(v,v)\in\D$ by transitivity}\\
  &\hgt(j)\le\M-1\ \text{is well defined}
    \why{previous line}
\end{align*}

\emph{(b), first claim.} Let $(j,l)\in\D$ and let $l_1,\dots,l_h$ be a chain above $j$.
\begin{align*}
  l_1,\dots,l_h,j &\ \text{is a chain above }l
    \why{$(l_h,j)\in\D$ by the chain, $(j,l)\in\D$ by assumption}\\
  \hgt(l) &\ \ge\ \hgt(j)+1
    \why{take $h=\hgt(j)$}\\
  \tier(j) &< \tier(l)
    \why{$\tier=1+\hgt$}
\end{align*}

\emph{(a).}
\begin{align*}
  &\text{the tiers partition }\Mset
    \why{they are the level sets of the function $\tier$}\\
  &\text{no two models of one tier are related by }\D
    \why{first claim of (b)}\\
  &\text{tiers }1,\dots,\max_j\tier(j)\text{ are nonempty}
    \why{downward induction from the second claim of (b), proved next from the first claim only}
\end{align*}

\emph{(b), second claim.} Let $\hgt(j)=h\ge1$, with a longest chain $l_1,\dots,l_h$ above $j$.
\begin{align*}
  \hgt(l_h) &\ \ge\ h-1
    \why{$l_1,\dots,l_{h-1}$ is a chain above $l_h$}\\
  \hgt(l_h) &\ \le\ h-1
    \why{$(l_h,j)\in\D$ and the first claim}\\
  \tier(l_h) &= \tier(j)-1,\ \text{and }l_h\text{ dominates }j
    \why{the two previous lines; $(l_h,j)\in\D$}
\end{align*}

\emph{(c).} Let $l_1,\dots,l_h$ be a longest chain above $j$, so $h=\hgt(j)$.
\begin{align*}
  (l_i,j) &\in\D\quad\text{for every }i
    \why{transitivity along the chain}\\
  L_{j,t}=1+\#\{l:(l,j)\in\D\} &\ \ge\ 1+h=\tier_t(j)
    \why{the $l_i$ are distinct}\\
  \rank_j(\thetastar) &\ \ge\ L_{j,t}
    \why{Theorem~\ref{thm:ranks}(b) applied to $\Dbar_t$, on $\good$}
\end{align*}
For the last assertion let $j$ have the largest tier, $s$; a longest chain above $j$, followed by
$j$, consists of $s$ models with strictly decreasing $\thetastar$ on $\good$.

\emph{(d).}
\begin{align*}
  \D_t\subseteq\D_{t+1} &\ \Rightarrow\ \Dbar_t\subseteq\Dbar_{t+1}
    \why{a directed path in $\D_t$ is one in $\D_{t+1}$}\\
  &\ \Rightarrow\ \text{every chain above }j\text{ at time }t\text{ is one at time }t+1
    \why{previous line and the definition of a chain}\\
  &\ \Rightarrow\ \tier_t(j)\le\tier_{t+1}(j)
    \why{$\hgt$ cannot decrease, and $\tier=1+\hgt$}
\end{align*}

\emph{(e).} Recall that an antichain is a set of models no two of which are related by $\D$.
Let $s$ be the number of tiers.
\begin{align*}
  &\text{there is a chain of }s\text{ elements}
    \why{a longest chain above a model of maximal tier, followed by that model}\\
  &\text{an antichain meets a chain in at most one element}
    \why{two elements of a chain are related by $\D$}\\
  &\text{every partition into antichains has at least }s\text{ blocks}
    \why{each block contains at most one element of that chain}\\
  &\text{the tiers attain }s
    \why{they are $s$ antichains by (a); \citet{mirsky1971}}
\end{align*}
\end{proof}

\begin{proof}[Proof of Proposition~\ref{prop:retire}]
Let $(\F_t)$ be the filtration of Section~\ref{sec:setup}, which contains the scores of all
models on the items drawn so far, observed or not. Under (F) the order $\Pi$ is drawn before
evaluation, so retirement does not change which item is revealed next.

\emph{Induction on $t$: $A_t$ is $\F_{t-1}$-measurable and $\F^{\mathrm{obs}}_t\subseteq\F_t$.}
\begin{align*}
  A_1 &= \Mset\ \text{is }\F_0\text{-measurable}
    \why{deterministic}\\
  A_s\ \text{is }\F_{s-1}\text{-measurable},\ s\le t
    &\ \Rightarrow\ \ind{j\in A_s},\ \ind{j\in A_s}X_{sj}\ \text{are }\F_t\text{-measurable},\ s\le t
    \why{induction hypothesis; $X_{sj}$ is $\F_s$-measurable and $\F_s\subseteq\F_t$}\\
  &\ \Rightarrow\ \F^{\mathrm{obs}}_t\subseteq\F_t
    \why{$\F^{\mathrm{obs}}_t$ is generated by $\xi$ and these variables}\\
  &\ \Rightarrow\ A_{t+1}\ \text{is }\F_t\text{-measurable}
    \why{$A_{t+1}$ is a measurable function of them and of $\xi$}
\end{align*}

\emph{Stopping times, bets, frozen wealths.}
\begin{align*}
  \{\sigma_j\le t\} &= \{j\notin A_{t+1}\}
    \why{the sets $A_t$ are nested, so $j\in A_s$ exactly for $s\le\sigma_j$}\\
  &\in\F^{\mathrm{obs}}_t
    \why{$A_{t+1}$ is a function of $\xi$ and of the scores observed on items $1,\dots,t$}\\
  \sigma_j\ \text{and}\ \sigma_j\wedge\sigma_l &\ \text{are }(\F_t)\text{-stopping times}
    \why{previous line makes them $(\F^{\mathrm{obs}}_t)$-stopping times, and $\F^{\mathrm{obs}}_t\subseteq\F_t$}\\
  \text{every }(\F^{\mathrm{obs}}_t)\text{-predictable bet} &\ \text{is }(\F_t)\text{-predictable}
    \why{$\F^{\mathrm{obs}}_{t-1}\subseteq\F_{t-1}$}\\
  (\Ep^{jl}_t) &\ \text{is a nonnegative }(\F_t)\text{-supermartingale started at }1
    \why{Proposition~\ref{prop:instances} with these bets, under any law with $\theta_j\le\theta_l$}\\
  (\Ep^{jl}_{t\wedge\sigma_j\wedge\sigma_l}) &\ \text{is a nonnegative }(\F_t)\text{-supermartingale started at }1
    \why{last sentence of Lemma~\ref{lem:betting} at the stopping time $\sigma_j\wedge\sigma_l$}
\end{align*}
Because the sets $A_t$ are nested, both models were scored on every item up to
$\sigma_j\wedge\sigma_l$, so $S_{t-1}$ in \eqref{eq:bfinite} is computable and the frozen process
depends on observed scores only. The analysis runs in the larger filtration; the procedure never
needs the unobserved scores. The frozen processes therefore satisfy Assumption~\ref{ass:eproc},
and every result of Section~\ref{sec:theory} and Appendix~\ref{app:theory} that rests on it
continues to hold.
\end{proof}

\begin{proof}[Proof of Corollary~\ref{cor:cost}]
Under the all-pairs rule a pair $\{j,l\}$ is frozen at $\sigma_j\wedge\sigma_l$, and $\sigma_l$
is the first time every pair at $l$, $\{j,l\}$ included, is resolved (likewise $\sigma_j$). This rule
depends only on observed scores, as Proposition~\ref{prop:retire} requires. So a
pair is frozen only after it is resolved. Until then both models are active, every $Z^{jl}_s$ is
observed, and the bet depends on the pair's own past, which is the same with and without
retirement. So until resolution the pair's wealths equal those computed without retirement. If
$\{j,l\}$ were unresolved at time $\tau^{B}_{q(j,l)}$ it would be unfrozen, its wealth in the
true direction would equal the unretired one, which has crossed $\M(\M-1)/\alpha$, and
e-Bonferroni $\subseteq\D_t$ would resolve it. Hence:
\begin{align*}
  \sigma_j &\le \max_{l\ne j}\tau^{B}_{q(j,l)}
    \why{every pair at $j$ is resolved by its own $\tau^{B}_{q(j,l)}$, as shown above}\\
  \Pr(\sigma_j>t) &\le \textstyle\sum_{l\ne j}\Pr(\tau^{B}_{q(j,l)}>t)\le(\M-1)e^{-t/\beta_j}\quad\text{for integer }t\ge\kappa_j
    \why{union bound; Proposition~\ref{prop:power}(c) at each pair, since $t\ge2\eta/\mu_q$ and $a_q\ge1/\beta_j$}\\
  \E[\sigma_j] &\le t_1+\textstyle\sum_{t\ge t_1}(\M-1)e^{-t/\beta_j}\le t_1+1+\beta_j,\quad t_1=\lceil\max\{\kappa_j,\beta_j\log(\M-1)\}\rceil
    \why{geometric series, $\sum_{t\ge t_1}e^{-t/\beta_j}\le e^{-t_1/\beta_j}(1+\beta_j)$, and $(\M-1)e^{-t_1/\beta_j}\le1$}\\
  \E[\sigma_j] &\le \kappa_j+\beta_j\bigl(\log(\M-1)+1\bigr)+2
    \why{$t_1\le\max\{\kappa_j,\beta_j\log(\M-1)\}+1$}
\end{align*}
which is \eqref{eq:cost}. Under the conditions of \eqref{eq:costgrid},
Proposition~\ref{prop:power}(d) gives, for every pair $q$ at $j$,
$2\eta/\mu_q\le11.2\rho_q\eta/\Delta_q^2\le11.2\rho_j\eta/\Delta_j^2$ and
$1/a_q\le\min\{74.7,270\rho_q\}/\Delta_q^2\le\min\{74.7,270\rho_j\}/\Delta_j^2$, since
$\rho_q\le\rho_j$, $\Delta_q\ge\Delta_j$ and $\min\{74.7,270\rho\}$ is nondecreasing in $\rho$;
summing over $j$ with $\min\{74.7,270\rho_j\}\le270\rho_j$ and $\log(\M-1)\le\eta$ gives the
order statement. On $\good$ no tied pair is ever certified (Theorem~\ref{thm:validity}), so a
model with a tied partner is never retired.
\end{proof}

\subsection{Integer programming and hardness of exact certification}

\begin{proof}[Proof of Proposition~\ref{prop:ilp}]
Identify $y\in\{0,1\}^{\pairs}$ with the relation $\{(a,b):y_{ab}=1\}$ on $\Mset$, and write
$[t]=\{1,\dots,t\}$ for the call times so far. Write $\ind{\T(\wo)}$ for the indicator vector of
$\T(\wo)$, whose entries are $\ind{(a,b)\in\T(\wo)}$ for $(a,b)\in\pairs$. A \emph{total
preorder} is a reflexive, transitive and total relation, that is, a weak order as defined in
Section~\ref{sec:method}.
\begin{align*}
  &y\ \text{satisfies totality and transitivity}\iff y=\ind{\T(\wo)}\ \text{for a unique weak order }\wo
    \why{a relation is a total preorder iff it is total and transitive; $\T(\wo)$ determines $\wo$}\\
  \Ep^\wo_s<1/\alpha &\iff \textstyle\sum_{(a,b)\in\T(\wo)}\Ep^{ab}_s<|\T(\wo)|/\alpha\iff g_s(y)<0
    \why{eq.~\eqref{eq:EW} with arithmetic-mean weights; $|\T(\wo)|=\sum_{ab}y_{ab}$}\\
  \wo\in\CS_t &\iff \max_{s\in[t]}g_s(y)<0
    \why{$\CS_t$ discards $\wo$ at the first call with $\Ep^\wo_s\ge1/\alpha$}\\
  (j,l)\in\D_t &\iff z^\star([t])\ge0
    \why{eq.~\eqref{eq:D}: no $\wo\in\CS_t$ has $(j,l)\in\T(\wo)$; the minimum in \eqref{eq:ilp} is over a finite set and is attained}\\
  z^\star(S) &\le z^\star([t])\quad\text{for every }S\subseteq[t]
    \why{for each $y$ a maximum over fewer times is no larger}
\end{align*}
The algorithm maintains three invariants at every call time $t$. First, every order in the pool
lies in $\CS_t$. This holds because an order enters the pool only after passing every survival
constraint and leaves at the first call at which it fails one, and $\CS_t$ is decreasing in $t$.
Second, a pair covered by the pool is not in $\D_t$, because the covering order has
$\max_{s\in[t]}g_s<0$. Third, a pair added to $\D$ is in $\D_t$, by
Proposition~\ref{prop:shortcut} when it is added through $B^{jl}_t$, and by the last two lines
above when it is added because $z^\star(S\cup\{t\})\ge0$. Every pair not in $\D$ is covered when the call
ends. Every pair in $\D$ stays in $\D_{t'}$ for $t'\ge t$ because $\D_{t'}$ is monotone, so never
revisiting a certified pair loses nothing. Hence the returned set is $\D_t$. Each pass through
the loop for one pair either stops or adds to $S$ a call time at which the minimizer $\wo$ has
$g_s\ge0$. That time is not in $S\cup\{t\}$, because $\max_{s\in S\cup\{t\}}g_s(y)=z^\star(S\cup\{t\})<0$ for
the minimizer. Each failed pass therefore adds a new time from $[t-1]$, so the loop stops after at
most $t$ solves.
\end{proof}

\begin{proof}[Proof of Proposition~\ref{prop:hard}]
The reduction is from the feedback arc set problem in tournaments. Given a tournament $G$ on
$[\M]$, one arc between every two vertices, and an integer $K$, the problem is to decide whether
some linear ordering of $[\M]$ has at most $K$ \emph{backward} arcs, arcs $a\to b$ with $b$
placed above $a$. This problem is NP-hard \citep{alon2006,charbit2007}. Write
$C=\binom{\M}{2}$. An ordering and its reverse have $C$ backward arcs between them, so some
ordering has at most $\lfloor C/2\rfloor$, and instances with $K\ge\lfloor C/2\rfloor$ are
trivial. The problem therefore stays NP-hard with $1\le K\le\lfloor C/2\rfloor-1$, which we
assume.

\emph{Construction.} Take $\M'=\M+2$ models, the vertices of $G$ and two more, $j$ and $l$. Let
$C'=\binom{\M'}{2}$, let $\nu$ be the smallest integer with $2^\nu\ge C'$, and set
$c=2^\nu(K+\tfrac12+\M)/C'$ and $\alpha=1/c$, which lies in $(0,1)$ since $2^\nu\ge C'$ and
$K+\tfrac12+\M>1$. The wealth matrix is: for vertices $a\ne b$, $\Ep^{ab}=2^\nu$ if $a\to b$ is an
arc of $G$ and $\Ep^{ab}=0$ otherwise; $\Ep^{jl}=0$ and $\Ep^{lj}=2^\nu$; and
$\Ep^{ja}=\Ep^{aj}=\Ep^{la}=\Ep^{al}=2^{\nu-1}$ for every vertex $a$. Recall that $(a,b)\in\T(\wo)$,
that is $a\wle_\wo b$, reads ``$b$ at least as good as $a$'', so an arc $a\to b$ contributes $2^\nu$
exactly when $\wo$ places $b$ at least as high as $a$. The instance has $\M'(\M'-1)$ wealth
entries in $\{0,2^{\nu-1},2^\nu\}$ with $2^\nu<2C'$ and a rational $\alpha$ with numerator and denominator
below $2C'(K+\M+1)$, all of $O(\log\M)$ bits, so it is constructed in polynomial time. We show
that $(j,l)\notin\D_1$ if and only if $G$ has an ordering with at most $K$ backward arcs. Since
the feedback arc set problem is NP-hard, deciding $(j,l)\in\D_1$ is coNP-hard.

\emph{If.} Let an ordering of $[\M]$ have $B\le K$ backward arcs and let $\wo$ rank the
vertices by it, all strictly, with $j$ strictly below every vertex and $l$ strictly above.
\begin{align*}
  \textstyle\sum_{(a,b)\in\T(\wo)}\Ep^{ab} &= 2^\nu B+\M\cdot2^{\nu-1}+\M\cdot2^{\nu-1}+0=2^\nu(B+\M)
    \why{backward arcs; one of $(j,a),(a,j)$ per vertex, likewise for $l$; $(j,l)\in\T(\wo)$ costs $\Ep^{jl}=0$}\\
  \Ep^\wo &= 2^\nu(B+\M)/C'\ \le\ 2^\nu(K+\M)/C'\ <\ c
    \why{$|\T(\wo)|=C'$, no ties; $B\le K$}
\end{align*}
so $\wo$ survives, $(j,l)\in\T(\wo)$, and $(j,l)\notin\D_1$.

\emph{Only if.} Let $\wo$ survive with $(j,l)\in\T(\wo)$. Ties are allowed, so let $B$ be the
number of backward arcs between vertices that $\wo$ orders strictly, $n_v$ the number of
tied pairs of vertices, and $n_{\mathrm{tie}}$ the number of tied pairs of all kinds, so that
$|\T(\wo)|=C'+n_{\mathrm{tie}}$. A tied pair of vertices costs $2^\nu+0$. A vertex strictly ordered against $j$
costs $2^{\nu-1}$ and one tied with $j$ costs $2^\nu$, likewise for $l$. The pair $\{j,l\}$ costs
$0$ if $j\wlt_\wo l$ and $2^\nu$ if the two are tied. Hence
\begin{align*}
  \textstyle\sum_{(a,b)\in\T(\wo)}\Ep^{ab} &\ \ge\ 2^\nu\bigl(B+n_v\bigr)+2^\nu\M+2^{\nu-1}(n_{\mathrm{tie}}-n_v)
    \why{every tie outside the vertex pairs adds at least $2^{\nu-1}$ to the strict baseline $2^\nu(B+\M)$}\\
  \frac{2^\nu(B+\M+\tfrac{n_{\mathrm{tie}}}2+\tfrac{n_v}2)}{C'+n_{\mathrm{tie}}} &< c=\frac{2^\nu(K+\tfrac12+\M)}{C'}
    \why{$\wo$ survives; $|\T(\wo)|=C'+n_{\mathrm{tie}}$}\\
  B+\tfrac{n_v}2+n_{\mathrm{tie}}\Bigl(\tfrac12-\tfrac{K+\frac12+\M}{C'}\Bigr) &< K+\tfrac12
    \why{multiply out and cancel $2^\nu\M$}\\
  K+\tfrac12+\M\ \le\ \tfrac{\M(\M-1)}4+\M-\tfrac12 &< \tfrac{(\M+2)(\M+1)}4=\tfrac{C'}2
    \why{$K\le\lfloor C/2\rfloor-1\le\M(\M-1)/4-1$}\\
  B+\tfrac{n_v}2 &< K+\tfrac12
    \why{the coefficient of $n_{\mathrm{tie}}$ is positive}
\end{align*}
Now refine $\wo$ to a linear ordering of the vertices: keep the strict comparisons and order
each class of tied vertices so that at most half of the arcs inside it are backward. This is
possible because an ordering of the class and its reverse have every internal arc backward
exactly once. The arcs inside the classes number $n_v$, so the ordering has at most
$B+n_v/2<K+\tfrac12$ backward arcs, hence at most $K$. Allowing ties therefore cannot help an
ordering survive.

\emph{Realizability.} With $N=\infty$ the factor of the update \eqref{eq:wealth} is $1+\lambda Z$
with $Z\in\{-1,0,1\}$ for binary scores, so with bets in $\{0,1\}$ every factor is $0$, $1$ or
$2$ and every wealth is $0$ or a power of two. Give each ordered pair $(a,b)$ its own items, on
which every other pair bets $0$. With bet $1$, $k$ items on which model $a$ scores $1$ and model
$b$ scores $0$ make $\Ep^{ab}=2^k$. With bet $1$, one item on which $a$ scores $0$ and $b$ scores
$1$ makes $\Ep^{ab}=0$. At most $(\nu+1)\M'(\M'-1)$ items are needed, the schedule is fixed
in advance, hence predictable, and the certifier is called once, after the last item. The
default mixture bet is not covered: its factors are at least $\tfrac12$, so its wealths never
vanish.
\end{proof}

\section{Scope of the guarantee and open problems}
\label{app:rigor}

Every result stated in this paper is proved in Appendix~\ref{app:proofs}. Four kinds of
machine check accompany the proofs but do not replace them. The counts of weak orders,
$13$ of $63$ intersections for $\M=3$ and $75$ of $4095$ for $\M=4$, are verified by exhaustive
enumeration, as is the three-model example of Appendix~\ref{app:theory}. That the exact test
certifies everything the full closure certifies is checked at $\M=3$ by enumerating all $63$ subsets for each of
many random wealth matrices. The remaining containments, between e-Bonferroni, the shortcut and
the exact test, are checked on random wealth matrices at $\M=3$, $4$ and $5$. The fourth averages the
finite-benchmark wealth exactly over all item orders and confirms that its expectation is at most
one at every time. These are implementation checks on finitely many inputs.

\paragraph{What the guarantee does not cover.}
The method needs a scalar ability per model, and pairwise-battle data do not supply one. In
battle data a platform picks two models, both answer the same prompt, and a judge picks the
winner. The natural parameter is then $p_{jl}=\Pr(j\text{ beats }l)$, and
``$p_{jl}>\tfrac12$'' is not transitive in general: $j$ may usually beat $l$, and $l$ usually
beat $m$, while $m$ usually beats $j$. Lemma~\ref{lem:closure} and the rank \eqref{eq:rank} rest
on transitivity of $>$ on $\R$. The method also needs the set of models fixed before any data
arrive, because a new entrant enlarges $\W$ and changes every $\T(\wo)$. Models arriving over time
would require online closed testing, which admits new hypotheses as they arrive
\citep{fischerramdas2024}, or an error budget split in advance among future entrants. Finally, the
method needs the item order not to depend on the data.
Lemma~\ref{lem:betting} allows a pair to be frozen at a stopping time under both sampling models,
but not the choice of \emph{which item} to evaluate next: under (S) that breaks the independence
of $X_t$ from $\F_{t-1}$, and under (F) it breaks the uniformity of $\Pi(t)$. The
inverse-probability-weighted signal of \citet{zhou2026celeus} is a natural repair and is not
worked out here. Pausing a pair and resuming it is covered under (F) only if the skipped items
are scored retroactively before its wealth is updated again, since $S_{t-1}$ in
\eqref{eq:bfinite} needs every revealed difference. Scoring the skipped items forgoes the saving
that pausing was meant to buy.

\paragraph{What is open inside it.}
Deciding membership in $\D_t$ is coNP-hard over the wealth matrices that the update
\eqref{eq:wealth} can produce with some predictable bets
(Proposition~\ref{prop:hard}), by a reduction that uses bets in $\{0,1\}$. The default mixture
bet keeps every wealth strictly positive, and whether it can produce such hard instances is
unknown. Nothing in this paper depends on the answer: Algorithm~\ref{alg:ilp} solved every
instance in our experiments, and Algorithm~\ref{alg:shortcut} is the polynomial fallback. On the power
side, Proposition~\ref{prop:power} and Corollary~\ref{cor:cost} give consistency, certification
times and the expected cost of retirement. That cost bound covers the all-pairs rule under (S).
It does not cover the top-$k$ rule behind the compute savings we report, under which a pair can
be frozen before it is certified. Validity under either rule follows from Proposition~\ref{prop:retire},
which needs no such bound. What remains open is a rate under (F), a cost bound for the top-$k$ rule, and a rate for
the exact test sharper than the bound it inherits from e-Bonferroni.

\section{Experiments}
\label{app:experiments}

The data, the designs and the definition of an error are common to every experiment and are
given first. Each subsection after that asks one question and answers it. Appendices~\ref{app:epeek}
to~\ref{app:eprice} carry the numbers reported in Section~\ref{sec:experiments}, in the order in
which it reports them, and Appendices~\ref{app:dependence} and~\ref{app:ecorr} add experiments
that Section~\ref{sec:experiments} does not report.

\paragraph{The real leaderboard.}
The data are the per-item results of $395$ Open LLM Leaderboard models, as released with
tinyBenchmarks \citep{maiapolo2024tiny,openllmleaderboard2023}, on six benchmarks: MMLU
\citep{hendrycks2021mmlu} ($14{,}042$ items), HellaSwag \citep{zellers2019hellaswag} ($10{,}042$),
GSM8K \citep{cobbe2021gsm8k} ($1319$), WinoGrande \citep{sakaguchi2020winogrande} ($1267$),
ARC-Challenge \citep{clark2018arc} ($1172$) and TruthfulQA \citep{lin2022truthfulqa} ($817$, whose per-item scores can take any value in $[0,1]$).
The estimand is (F): the truth is the ranking by mean score on the whole benchmark, so every
certified dominance can be checked against it. Three model subsets are used, the $20$ and the $8$
highest-scoring models by mean score over the six benchmarks and the \emph{spread subset} of $20$ models evenly spaced
through the ranking of all $395$, together with all $395$ released models with no further selection.
Items are revealed in a uniformly random order, redrawn $50$ times per setting, and the procedure
reports every $1\%$ of the benchmark, with checkpoints at fractions $0.05$, $0.10$, $0.25$,
$0.50$, $0.75$ and $1.00$. Throughout, $\alpha=0.05$, the bets are the mixture grid of
Section~\ref{sec:method}, and the width of a rank interval is $U_{j,t}-L_{j,t}$, zero for a model
whose rank is resolved. Certification uses the shortcut for the twenty-model subsets and for
all $395$, and enumeration for the eight-model subset.

\paragraph{The simulations.}
Two simulated designs recur. The first streams i.i.d.\ items under (S), and the models share each
item's difficulty, so that they are positively dependent. Item $t$ has difficulty
$d_t\sim N(0,1)$, and model $j$, with logistic ability $\omega_j$, answers it correctly with
probability $1/(1+e^{-(\omega_j-d_t)})$, independently of the other models given $d_t$. Its six
models have logistic abilities $(0.5,\dots,0.5)$, all tied at accuracy $0.602$; $(1,1,0.6,0.6,0.2,-0.2)$,
two tied pairs; $(1,0.98,0.62,0.60,0.22,0.20)$, near ties; $(1,0.8,0.6,0.4,0.2,0)$, evenly
spread; or, in the bet comparison of Appendix~\ref{app:ecorr}, $(0.5,0.4,0.3,0.2,0.1,0)$, a close
race with gaps of about $0.02$. True accuracies lie between $0.46$ and $0.70$. This design is used
wherever a known tie is needed, because ties are where a multiple-comparison procedure is tested
hardest. The second design is (F). It draws one benchmark of $N=5000$ items from the same
logistic model, and the truth is each model's accuracy on those items. It has either $20$ models with accuracies from $0.40$ to $0.78$ and
adjacent gaps between $0.016$ and $0.022$, or $6$ models from $0.50$ to $0.70$ with gaps of about
$0.04$, or $6$ models whose two leaders have equal logistic ability.
Replication counts are given with each experiment, and shaded bands in the figures are two
standard errors.

\paragraph{What counts as an error.}
A run counts as wrong if any report it issues, at any monitored time, certifies a dominance that
is false. A rank interval that misses a true rank implies such a dominance, by
Theorem~\ref{thm:ranks}(b), so this single count covers the dominances and the intervals
together; the exact rank sets $\Rset_j$ of \eqref{eq:D} are not scored separately. It is the error rate faced by
an analyst who is free to stop at any look, which is the quantity the guarantee bounds.

\paragraph{The fixed-sample competitors.}
Three constructions stand in for current practice. Holm rank sets apply Holm's step-down
correction \citep{holm1979} to the $\M(\M-1)$ one-sided paired tests and convert the rejections
into rank intervals by the counting formula \eqref{eq:LU}. A \emph{step-down} correction tests the
hypotheses in order of their evidence and relaxes the threshold as hypotheses are rejected. Holm rank
sets are the simultaneous rank intervals of \citet{neuhof2024confident}, which
\citet{neuhof2026mmlu} also compute for a leaderboard. Here the paired test is either a paired
$z$-test or an exact one-sided McNemar test \citep{mcnemar1947}, which is a sign test on the
items where the two models' scores differ. Romano--Wolf step-down rank
sets \citep{romano2005,mogstad2024} set the threshold by a bootstrap of the studentized
maximum, the largest over pairs of the mean paired difference divided by its standard error. The
bootstrap is a centered multiplier bootstrap with $200$ draws, which reweights the centered
per-item differences by independent random multipliers. All are recomputed at each look where
they are monitored.

\paragraph{Implementation.}
The implementation, a few hundred lines of Python, follows Section~\ref{sec:method} alone. Its
automated tests run the machine checks listed in Appendix~\ref{app:rigor}. A separate simulation shows that the bound of
Theorem~\ref{thm:validity} is nearly attained: with four tied models of accuracy $0.5$ scored
independently, $400$ i.i.d.\ items under (S) checked after every item, a fixed bet of $0.5$ and
$\alpha=0.2$, the true weak order is rejected at some time in $18.4\%$ of $4000$ runs, with a
standard error of $0.6\%$.

\subsection{Does monitoring break a fixed-sample rank set, and does ours survive it?}
\label{app:epeek}

\paragraph{Setup.}
Six models of the first simulated design, all tied, with tied pairs at ranks 1 and 3, or with near ties, are evaluated on $2000$ items, and every procedure
sees the same data. The fixed-sample rank sets are recomputed at each of a number of equally
spaced looks, which is the horizontal axis of Figure~\ref{fig:peeking}. Ours is checked at the same
looks, although its guarantee would also cover a check after every item. There are $5000$ replications.

\paragraph{Monitoring destroys the guarantee of a fixed-sample rank set and leaves ours intact.}
At $200$
looks the fixed-sample sets make a false statement in $30\%$ of runs with the exact McNemar test and
$51\%$ with the paired $z$-test, against $1.6\%$ for our exact test. At a single pre-planned look the same two
procedures are valid, at $3.3\%$ and $3.6\%$, so what fails is the repeated look and not the
normal approximation. The effect does not need exact ties. With tied pairs at ranks 1 and 3 the
figures are $9.7\%$ and $18\%$ against our $0.3\%$, and when no two models are exactly tied but
some gaps are $0.004$ they are $4.6\%$ and $11.7\%$ against our $0.2\%$.

\subsection{How much of a real leaderboard can be certified?}
\label{app:ereal}

\paragraph{Setup.}
The real leaderboard described above, in two forms: the three selected subsets of
twenty and eight models, and all $395$ released models with no further selection. Selecting
the twenty highest-scoring models is legitimate under (F), because the score matrix is fixed
before the item order is drawn, but it is not a selection an operator could make prospectively,
which is why the unselected run matters.

\paragraph{Half a benchmark certifies up to half of the true dominances among close models, and
$26\%$ to $86\%$ of them among all $395$ models.}
At half of MMLU the twenty highest-scoring models
fall into $3.0$ certified tiers with $50\%$ of the $190$ true dominances certified and a mean rank-interval
width of $9.5$, while the spread subset reaches $11.2$ tiers. At three quarters of MMLU ours certifies
$63\%$ of the top-twenty true dominances against $51\%$ for Holm and $54\%$ for Romano--Wolf, with
$4.9$ tiers against $3.3$ and $3.6$. Across the $900$ runs that the three subsets, six benchmarks
and $50$ orders make up, exactly one report was false, in one item order of the HellaSwag
eight-model set, well below $\alpha=5\%$.

\paragraph{The whole leaderboard.}
Removing the selection makes certification easier, even though the correction grows.
With all $395$ models,
$86\%$ of the $77{,}802$ true dominances of MMLU are certified at half the benchmark and $99.6\%$
at the end, against $83\%$ and $88\%$ for Holm, and the certified partial order has $225$ tiers
against Holm's $26$. Table~\ref{tab:full} gives the six benchmarks; WinoGrande is the hard case,
at $26\%$ of true dominances at half. Ours never certified a false dominance in $300$ runs monitored at about a hundred looks
each, and neither did the Holm sets, which use the paired $z$-test and are recomputed at the six
checkpoints. With $395$ models the shortcut threshold is $\M(\M-1)/\alpha=3{,}112{,}600$, but the correction
grows only as $2\log\M$ while the number of easy pairs grows as $\M^2$. These percentages come
from the shortcut, so by Proposition~\ref{prop:shortcut} the exact test would certify at least as many. Romano--Wolf is omitted
here because its bootstrap over $155{,}630$ pairs is not affordable at this $\M$.

\begin{table}[t]
\caption{All $395$ released models, no further selection, $50$ item orders per benchmark. True dominances
certified, as a percentage of those that hold on the whole benchmark; tiers at the end. Holm uses
the paired $z$-test.}
\label{tab:full}
\centering\small
\begin{tabular}{@{}lrrrrrrrr@{}}
\toprule
 & & \multicolumn{2}{c}{$0.25$} & \multicolumn{2}{c}{$0.50$} & \multicolumn{2}{c}{$1.00$} & tiers\\
\cmidrule(lr){3-4}\cmidrule(lr){5-6}\cmidrule(lr){7-8}\cmidrule(lr){9-9}
Benchmark & items & ours & Holm & ours & Holm & ours & Holm & ours / Holm\\
\midrule
MMLU          & $14{,}042$ & $76$ & $77$ & $86$ & $83$ & $99.6$ & $88$ & $225$ / $26$\\
HellaSwag     & $10{,}042$ & $62$ & $63$ & $77$ & $74$ & $99.2$ & $81$ & $165$ / $24$\\
GSM8K         & $1319$     & $63$ & $66$ & $78$ & $75$ & $98.9$ & $83$ & $124$ / $14$\\
WinoGrande    & $1267$     & $5$  & $6$  & $26$ & $17$ & $95.5$ & $34$ & $53$ / $5$\\
ARC-Challenge & $1172$     & $18$ & $22$ & $48$ & $39$ & $97.0$ & $56$ & $68$ / $8$\\
TruthfulQA    & $817$      & $19$ & $35$ & $55$ & $53$ & $97.0$ & $66$ & $69$ / $13$\\
\bottomrule
\end{tabular}
\end{table}

\paragraph{Why the end of a benchmark flatters our sets.}
The large gap at fraction $1.0$ in Figure~\ref{fig:real} is not a power difference. Under (F) a
pair is resolved with certainty once the unrevealed items can no longer overturn it, and the
offset \eqref{eq:bfinite} captures most of this. The fixed-sample constructions instead treat the
benchmark as a sample. They therefore certify only $55$--$58\%$ of the top-twenty true dominances on
MMLU at the point where the ranking is in fact known exactly. Our own sets fall short of the full
ranking for pairs whose margin on the whole benchmark is one or two items. The bet uses only the
null's bound on the conditional mean of the next difference. A null that has become impossible,
because the revealed items already settle the pair, is therefore not by itself infinite evidence.
After that point the offset sits at its lower clip $-1+\delta$, with $\delta=0.01$, and each item
multiplies the wealth by at most $1+\lambda(2-\delta)/\delta$. Pairs with a margin of one or two
items have very few such items. The fair comparison is therefore at fractions up to one half,
where the three procedures are close and the fixed-sample sets are ahead early on.

\subsection{How much compute does retirement save?}
\label{app:ecost}

\paragraph{Setup.}
Cost is counted in model-item evaluations against a fixed benchmark of $N=5000$ items, divided by
what a full evaluation of every model would cost, with $300$ replications and reports every $25$
items. Five retirement rules are compared: no retirement; stopping everything once every model's
top-$k$ status is certified; retiring a model once its own top-$k$ status is certified, which is
the rule of Section~\ref{sec:method}; the same but only once every pair between the model and a
model whose status is still open is certified, so that a retirement can never block another
certification; and, for the full ranking, retiring a model once every pair involving it is
certified. The group-sequential comparison is a Pocock design \citep{pocock1977}, as used by
\citet{arviv2026stop}, who add items in batches of $100$. Our implementation compares a
$z$-statistic with one constant critical value at each of ten equally spaced looks fixed in
advance, without a multiplicity correction, and a second version adds a Bonferroni correction
over the $380$ ordered pairs. Both versions close their certifications under transitivity.

\paragraph{Retiring each model as soon as its own question is answered costs about a fifth of a
full evaluation, and stopping everything at one moment costs four times as much.}
Certifying the
top-$3$ status of all $20$ simulated models costs $22\%$ of a full run under the retirement rule
against $84\%$ under the common stopping rule, with no report ever wrong. With six models and
$k=2$ the three top-$k$ rules cost $44\%$, $22\%$ and $30\%$: the cautious fourth rule gained
nothing, since the plain rule already reached the goal in every run, as did no retirement at all. Certifying the entire ranking is far more expensive, $51\%$ of a full run for
six models and $90\%$ for twenty, because adjacent gaps of about $0.02$ are resolved only when
the finite benchmark is nearly exhausted. The mixture bet cost less than a fixed bet of $0.25$ in
every setting but one, for example $51\%$ against $63\%$ for the full ranking of six models. The
exception is two leaders of equal logistic ability among six models, where the fixed bet was
slightly cheaper under the three top-$k$ rules.

\paragraph{The cost barely depends on $k$.}
Between $k=2$ and $k=16$ among twenty close models the cost stays at $21$--$24\%$ of a full run,
because it is set by the models adjacent to the $k$th place and not by $k$ itself. At $k=1$
and $k=19$ it falls to $18\%$ and $16\%$, where only one side of the boundary has to be resolved.
Ours reaches the goal in $97.5$--$99.5\%$ of runs at every $k$, over $200$ replications.

\paragraph{On the real leaderboard.}
The top-$3$ goal costs $31\%$ of a full MMLU run for the twenty highest-scoring models and $4\%$
for the spread subset, and $57\%$ of GSM8K and $68\%$ of HellaSwag for
the top twenty, reached in $96\%$ to $100\%$ of orders. Where the top twenty are
indistinguishable, on ARC-Challenge and WinoGrande, retirement saves almost nothing, as
expected. Table~\ref{tab:fullcost} shows that applying the same rule to all $395$ models costs a smaller
fraction still, because most models on the leaderboard are far from the top-$3$
boundary and retire almost at once, whereas the twenty highest-scoring are near-ties by
construction. The comparison is of fractions and not of compute: $6.2\%$ of the $395$-model
evaluation of MMLU is about four times the absolute cost of $31\%$ of the twenty-model one. Every
model's status is settled in $100\%$ of orders on MMLU and HellaSwag. On the other four
benchmarks some models near the boundary stay unsettled, about two on ARC-Challenge and
WinoGrande, three on GSM8K and nine on TruthfulQA. The all-or-nothing goal is therefore reached
in only $0$--$2\%$ of orders, while $97.7$--$99.5\%$ of models are settled. No report was wrong in
any of these runs.

\begin{table}[t]
\caption{Certifying every model's top-$3$ status, all $395$ models, $50$ item orders per
benchmark. Cost is model-item evaluations divided by $395N$; settled is the mean fraction of models whose top-$3$ status is certified at the end.}
\label{tab:fullcost}
\centering\small
\begin{tabular}{@{}lrr@{}}
\toprule
Benchmark & cost & settled\\
\midrule
MMLU          & $6.2\%$  & $100.0\%$\\
HellaSwag     & $15.2\%$ & $100.0\%$\\
GSM8K         & $16.7\%$ & $99.2\%$\\
WinoGrande    & $47.3\%$ & $99.5\%$\\
ARC-Challenge & $36.9\%$ & $99.5\%$\\
TruthfulQA    & $27.3\%$ & $97.7\%$\\
\bottomrule
\end{tabular}
\end{table}

\paragraph{Against a group-sequential design.}
The Pocock design is as cheap as ours and answers the question far less often. With the twenty
simulated models it reaches the goal in $56\%$ of runs against our $99\%$, because a $z$-test
cannot use the fact that the benchmark is finite and its looks must be fixed in advance. With no
multiplicity correction and only asymptotic validity, it also errs in $2.7\%$ of runs at that
cost, and in $3.7\%$ without retirement. The same holds across the sweep over $k$, where the
uncorrected design reaches the goal in $56$--$68\%$ of runs with a false statement in $3$--$4\%$.
A Bonferroni correction supplies the missing multiplicity correction but removes the saving. The
corrected version costs $32\%$ and reaches the goal in $3\%$ of runs. With two leaders of equal
logistic ability among six models, whose accuracies on the realized benchmark differ only by chance, the
uncorrected design errs in $1\%$ of runs, and in $2\%$ without retirement, while ours errs in
none, at an equal cost of $12\%$. All costs are averaged over every run, whether the goal was
reached or not.

\subsection{What does anytime validity cost?}
\label{app:eprice}

\paragraph{Setup.}
The comparison is made at the one look where a fixed-sample procedure is valid, so that it
measures only the power given up for the right to monitor. In simulation the first design runs
to $2000$ items, with either Holm construction run once and our exact test monitored at $200$
looks. On the real leaderboard the three fixed-sample constructions are compared with ours on the
twenty highest-scoring models, with Romano--Wolf computed at the six checkpoints only, since its
bootstrap is too slow for every look.

\paragraph{In simulation the price of anytime validity is a fraction of one dominance.}
After $2000$ items with two tied pairs, both
Holm constructions certify $12.98$ of $13$ true dominances and ours $12.79$; with near ties the
figures are $11.97$ and $11.62$ of $15$. On the real leaderboard the fixed-sample sets lead by up to eight
percentage points at a quarter of a benchmark, and by half that lead is gone: ours matches or
exceeds Romano--Wolf on five of the six benchmarks and is behind by four points on TruthfulQA. Table~\ref{tab:realE4}
shows Romano--Wolf to be the strongest of the three. Beyond one half ours pulls ahead for a reason unrelated to
power, given in Appendix~\ref{app:ereal}. This cost buys the right to look: the same
fixed-sample sets, once monitored, err in up to $4\%$ of orders (McNemar on TruthfulQA aside),
and ours in none. The
monitoring penalty is smaller here than in the tied simulations because the gaps among the top
twenty are mostly nonzero.

\begin{table}[t]
\caption{Fixed-sample constructions against ours, twenty highest-scoring models, $50$ item
orders. Top: true dominances certified at one half, as a percentage of those that exist. Bottom:
probability that some report, at any look, was false.}
\label{tab:realE4}
\centering\small
\begin{tabular}{@{}lrrrrrr@{}}
\toprule
 & MMLU & HellaSwag & GSM8K & WinoGrande & ARC-Challenge & TruthfulQA$^{a}$\\
\midrule
\multicolumn{7}{@{}l}{\emph{True dominances certified at one half of the benchmark (\%)}}\\
Ours, monitored        & 50 & 49 & 22 & 2 & 0 & 17\\
Holm, McNemar          & 47 & 47 & 19 & 1 & 0 & 61$^{b}$\\
Holm, paired $z$-test  & 48 & 48 & 20 & 1 & 0 & 16\\
Romano--Wolf           & 49 & 49 & 22 & 1 & 0 & 21\\
\midrule
\multicolumn{7}{@{}l}{\emph{Probability of a false statement at some look}}\\
Ours, monitored        & 0 & 0 & 0 & 0 & 0 & 0\\
Holm, McNemar          & 0 & 0.02 & 0 & 0 & 0 & 1.00$^{b}$\\
Holm, paired $z$-test  & 0 & 0.02 & 0.04 & 0 & 0.02 & 0\\
Romano--Wolf           & 0.02 & 0.02 & 0.02 & 0.04 & 0.02 & 0\\
\bottomrule
\end{tabular}

\smallskip
\raggedright\footnotesize
$^{a}$Per-item scores in $[0,1]$, not binary. $^{b}$The sign test of McNemar's procedure
targets the median of the paired differences, which is not the mean-score estimand for non-binary
scores. Its certifications on TruthfulQA are not comparable and every order contained a false
one. This is a test mismatch, not a monitoring effect.
\end{table}

\subsection{Does the guarantee depend on how model errors are correlated?}
\label{app:dependence}

\paragraph{Setup.}
Each regime holds every model's marginal accuracy fixed and varies only the dependence among its
item-level scores, so that any difference between regimes is attributable to dependence alone.
For six models and $4000$ i.i.d.\ items under (S), with the accuracies of the first simulated
design as marginals, binary scores come from a Gaussian copula,
$X_{tj}=\ind{a_j\sqrt{\gamma}\,G_t+\sqrt{1-\gamma}\,\varepsilon_{tj}\le\Phi^{-1}(\theta_j)}$ with
$G_t$ and $\varepsilon_{tj}$ independent standard normal, $\Phi$ the standard normal distribution
function, $\gamma\in[0,1)$ the weight of the shared factor $G_t$ and $a_j\in\{-1,1\}$ the sign with
which model $j$ loads on it. This gives $\Pr(X_{tj}=1)=\theta_j$ for every $\gamma$ and every
choice of signs. The four regimes are independence ($\gamma=0$), positive equicorrelation at
$\gamma=0.5$ and $\gamma=0.9$ with every $a_j=1$, and a mixed-sign regime at $\gamma=0.8$ in which
three models load positively and three negatively; the resulting binary correlations are about
$0$, $0.32$, $0.70$ and $\pm0.55$. Monitoring is every $20$ items, with the mixture bet and the exact
certifier, $5000$ replications for the tied configurations and $2000$ for the others, and the
same code runs unchanged in every regime.

\paragraph{Validity is unchanged across dependence structures and speed is not.}
With six tied models the
anytime false-statement probability is $2.0\%$, $1.5\%$, $1.6\%$ and $1.6\%$ across the four
regimes, and with two tied pairs it is $0.6\%$, $0.4\%$, $0.4\%$ and $0.5\%$, all at
$\alpha=0.05$. Table~\ref{tab:dependence} collects the four regimes. Speed depends strongly on the regime. With evenly spread
logistic abilities at half the items the procedure certifies $70\%$ of the true dominances under
independence, with a mean rank-interval width of $1.49$ and $3.3$ tiers, against $96\%$, $0.22$
and $5.3$ under strong positive dependence, where certifying every model's top-$3$ status costs
$16\%$ of a full evaluation instead of $41\%$. The mixed-sign regime is less favorable by design,
since pairs within a group benefit from positive dependence while pairs across groups are
negatively dependent and harder to resolve: it certifies more true dominances than independence at the
halfway point but its top-$3$ cost is $46\%$ and the goal is reached in only $26\%$ of runs. With
near ties, independence and the two positive regimes certify $74\%$, $79\%$ and $80\%$ at half.
Here $80\%$ is a practical ceiling, because the three near-tied pairs, with gaps of about
$0.004$, are almost never resolved within $4000$ items. The mixed-sign regime is slightly slower than
independence, at $72\%$. The regimes differ in the variance of the paired difference,
$\operatorname{Var}(Z^{jl}_t)=\operatorname{Var}(X_{tj})+\operatorname{Var}(X_{tl})-2\operatorname{Cov}(X_{tj},X_{tl})$:
positive covariance makes the paired signal less noisy and negative covariance more noisy, so a
pair with margin $0.075$ has paired-difference variance about $0.45$, $0.30$ and $0.14$ under
independence, $\gamma=0.5$ and $\gamma=0.9$, and its median first-certification time falls from
$28\%$ to $18\%$ to $9\%$ of the items.

\begin{table}[t]
\caption{Four dependence regimes with identical marginal accuracies. ``Anytime error'' uses six
tied models; the other columns use evenly spread logistic abilities after half the items, with the
probability of reaching the top-$3$ goal in parentheses.}
\label{tab:dependence}
\centering\small
\begin{tabular}{@{}lccccc@{}}
\toprule
Dependence & Anytime error & True dominances & Mean rank width & Tiers & Top-$3$ cost (goal)\\
\midrule
Independent              & 2.0\% & 70\% & 1.49 & 3.3 & 41\% (49\%)\\
Positive, $\gamma=0.5$   & 1.5\% & 76\% & 1.18 & 3.7 & 32\% (79\%)\\
Positive, $\gamma=0.9$   & 1.6\% & 96\% & 0.22 & 5.3 & 16\% (100\%)\\
Mixed sign, $\gamma=0.8$ & 1.6\% & 80\% & 1.01 & 3.9 & 46\% (26\%)\\
\bottomrule
\end{tabular}
\end{table}

\subsection{What does the multiplicity correction contribute?}
\label{app:ecorr}

\paragraph{Setup.}
Every certifier receives the same wealth paths, so only the multiplicity correction
varies. Four certifiers are compared: e-Bonferroni with closure,
which is the certification logic of \citet{gu2026serpant}; the shortcut of
Algorithm~\ref{alg:shortcut}; the exact test of Algorithm~\ref{alg:exact}; and a variant of the
exact test whose weights in \eqref{eq:EW} are uniform over the pairs of $\T(\wo)$ in the same or
in adjacent levels of $\wo$, not over all of them. The simulated comparison uses six
models over $4000$ items monitored every $20$, with $2000$ replications; every error rate was at
most $1.9\%$. The real comparison reuses the wealths of Appendix~\ref{app:ereal} unchanged.

\paragraph{The exact test certifies more pairs than e-Bonferroni, mostly early in evaluation,
and never fewer.}
On simulated data
the exact test certifies about $7\%$ more true dominances than e-Bonferroni after $500$ items and
at most $3\%$ more by $4000$, because a threshold enters the sample size an e-process needs only
through its logarithm: halving the multiplicity factor saves $\log 2$ divided by the growth rate,
a fixed number of items. The counts are $5.10$ against $4.77$ of $15$ true dominances after $500$
items with near ties and $5.81$ against $5.45$ with evenly spread logistic abilities, becoming $12.01$
against $12.00$ and $13.23$ against $12.90$, respectively, by $4000$. On the real wealths the exact test
certifies $2$--$5\%$ more than e-Bonferroni among the eight highest-scoring models up to three
quarters of a benchmark and nearly the same number at the end, ahead at $20\%$ of the looks and behind
at none. Table~\ref{tab:realE3} gives the four certifiers side by side. The shortcut and e-Bonferroni with closure are nearly identical in
these experiments. Here transitivity pooling guarantees that nothing is lost but brings no
measurable gain in power. Over $600$ further simulated runs the shortcut was ahead at $121$ of $120{,}000$
looks and behind at none. Adjacent-level weights are valid by Remark~\ref{rem:weights} but lose
most of the power early, certifying $0.3$ and $0.4$ pairs at $500$ items, because for a strict
order they average only the $\M-1$ adjacent pairs and give zero weight to the non-adjacent pairs,
whose large wealths are what eliminates orderings far from the truth.

\begin{table}[t]
\caption{True dominances certified from identical wealths, so that only the correction differs;
means over $50$ item orders and the six benchmarks, and for twenty models also over the top and
spread subsets. The last two columns count the looks at which the row beat or trailed the
reference.}
\label{tab:realE3}
\centering\footnotesize\setlength{\tabcolsep}{4pt}
\begin{tabular}{@{}lrrrrrrrr@{}}
\toprule
 & \multicolumn{6}{c}{fraction of the benchmark} & \multicolumn{2}{c}{looks vs.\ e-Bonferroni}\\
\cmidrule(lr){2-7}\cmidrule(lr){8-9}
 & $0.05$ & $0.10$ & $0.25$ & $0.50$ & $0.75$ & $1.00$ & ahead & behind\\
\midrule
\multicolumn{9}{@{}l}{\emph{Eight highest-scoring models, $27.5$ true dominances, $31{,}600$ looks}}\\
e-Bonferroni with closure (reference) & 2.15 & 3.11 & 4.86 & 7.51 & 12.62 & 26.04 & \multicolumn{2}{c}{reference}\\
Shortcut (Algorithm~\ref{alg:shortcut})  & 2.15 & 3.11 & 4.86 & 7.51 & 12.63 & 26.04 & 118 & 0\\
Exact (Algorithm~\ref{alg:exact})        & 2.19 & 3.21 & 4.99 & 7.90 & 13.02 & 26.11 & 6{,}320 & 0\\
Exact, adjacent-level weights  & 2.04 & 2.35 & 3.90 & 4.97 & 7.87 & 26.20 & 831 & 15{,}052\\
\midrule
\multicolumn{9}{@{}l}{\emph{Twenty models, top and spread subsets pooled, $187.8$ true dominances, $63{,}200$ looks}}\\
e-Bonferroni with closure (reference) & 27.0 & 37.9 & 64.5 & 89.5 & 110.7 & 177.6 & \multicolumn{2}{c}{reference}\\
Shortcut (Algorithm~\ref{alg:shortcut})  & 27.0 & 37.9 & 64.5 & 89.5 & 110.7 & 177.6 & 473 & 0\\
\bottomrule
\end{tabular}
\end{table}

The exact test and its adjacent-level variant enumerate every weak order, $545{,}835$ of them
for eight models but about $2.7\times10^{21}$ for twenty, which is why only the two closed-form
certifiers appear in the lower half of Table~\ref{tab:realE3}.

\paragraph{Beyond eight models, by integer programming.}
Computed by integer programming, the exact test certifies only slightly more than the shortcut. On the
twelve real settings at $\M=20$ the shortcut trails Algorithm~\ref{alg:ilp} by $1.4$ true
dominances at half a benchmark, with a range of $0.1$ to $2.9$, and by $0.3$ at the end.
Table~\ref{tab:ilp} reports every setting. The shortcut is never ahead, as Proposition~\ref{prop:shortcut} requires. The gap is
largest while certification is in progress, as on the spread subsets at one half, and nearly
vanishes once the shortcut has certified nearly everything, as at the end of every spread subset,
or while nothing is certified yet, as on ARC-Challenge and WinoGrande among the top twenty. Neither certifier was ever wrong, and Algorithm~\ref{alg:ilp} solved every
instance. Cost is modest: a run with looks every $1\%$ of a benchmark needs about
$294$ integer programs at $26$\,ms each, eight seconds in all, on one CPU core with HiGHS
through SciPy and the relative optimality gap set to zero. On simulated runs of $3000$ items with
looks every $20$, the times are about $0.2$, $15$ and $63$ seconds at $\M=10$, $20$ and $30$,
with $85$, $388$ and $935$ programs; there the exact test and the shortcut certify the same
$44.7$ of $45$, $189.3$ of $190$ and $431.3$ of $435$ true dominances.

\begin{table}[t]
\caption{The exact test by integer programming against the shortcut, at $\M=20$ with $50$ item
orders per row, drawn independently of those behind Table~\ref{tab:scale}, so the shortcut columns
differ slightly from it. True dominances certified, rank width at one half, and what the exact test cost
to run.}
\label{tab:ilp}
\centering\footnotesize\setlength{\tabcolsep}{2.5pt}
\begin{tabular}{@{}lrrrrrrrrrrr@{}}
\toprule
 & & \multicolumn{2}{c}{$0.25$} & \multicolumn{2}{c}{$0.50$} & \multicolumn{2}{c}{$1.00$} & \multicolumn{2}{c}{width at $0.50$} & \multicolumn{2}{c}{cost}\\
\cmidrule(lr){3-4}\cmidrule(lr){5-6}\cmidrule(lr){7-8}\cmidrule(lr){9-10}\cmidrule(lr){11-12}
Benchmark, subset & true & exact & shortcut & exact & shortcut & exact & shortcut & exact & shortcut & programs & s\\
\midrule
MMLU, top & 190 & 84.1 & 84.1 & 93.6 & 92.4 & 187.3 & 187.1 & 9.64 & 9.76 & 244 & 6\\
HellaSwag, top & 190 & 77.4 & 76.3 & 94.9 & 94.1 & 181.7 & 181.3 & 9.51 & 9.59 & 292 & 8\\
GSM8K, top & 190 & 14.4 & 13.4 & 45.1 & 43.1 & 183.2 & 182.8 & 14.49 & 14.69 & 331 & 9\\
WinoGrande, top & 180 & 0.1 & 0.1 & 3.3 & 3.1 & 153.9 & 153.2 & 18.67 & 18.69 & 238 & 6\\
ARC-Challenge, top & 179 & 0.0 & 0.0 & 0.3 & 0.2 & 148.2 & 147.1 & 18.97 & 18.98 & 213 & 6\\
TruthfulQA, top & 190 & 1.0 & 0.6 & 33.2 & 30.9 & 156.0 & 155.6 & 15.68 & 15.91 & 285 & 7\\
MMLU, spread & 190 & 165.5 & 164.7 & 173.3 & 172.9 & 189.6 & 189.5 & 1.67 & 1.71 & 217 & 6\\
HellaSwag, spread & 190 & 118.7 & 116.8 & 146.9 & 145.5 & 187.5 & 187.5 & 4.31 & 4.45 & 304 & 8\\
GSM8K, spread & 190 & 139.2 & 136.9 & 162.7 & 161.3 & 189.4 & 189.3 & 2.73 & 2.87 & 308 & 8\\
WinoGrande, spread & 186 & 42.9 & 40.6 & 92.6 & 89.8 & 184.0 & 183.8 & 9.74 & 10.02 & 360 & 10\\
ARC-Challenge, spread & 189 & 70.2 & 67.2 & 116.3 & 114.3 & 188.8 & 188.8 & 7.37 & 7.57 & 358 & 9\\
TruthfulQA, spread & 190 & 79.6 & 76.3 & 124.7 & 122.6 & 185.7 & 185.4 & 6.53 & 6.74 & 373 & 10\\
\bottomrule
\end{tabular}
\end{table}

\paragraph{The bet matters more than the correction.}
A bet with negative growth can keep a pair from ever being certified, whatever the correction. A fixed
bet of $0.25$ never certifies pairs whose gap is about $0.02$: with six models at accuracies
$0.50$ to $0.60$ it certifies $6.5$ of $15$ true dominances after $5000$ items and $7.0$ after
$20{,}000$, against $14.0$ for the mixture bet, $13.5$ for the plug-in bet and $13.0$ for online
Newton steps. This is the negative-growth effect of Remark~\ref{rem:bets}, which no multiplicity
correction can undo. The price of adapting is paid early and only where the fixed bet happens to
be well tuned: with accuracies evenly spaced from $0.50$ to $0.70$, where $0.25$ is close to the
log-optimal bet $\lambda^\star=\Delta/\rho$, the fixed bet leads after $500$ items, $6.9$ pairs
against $5.8$, but the mixture overtakes it by $2000$ and leads clearly by $4000$, $13.2$ against
$11.1$. The mixture was the best adaptive bet at every checkpoint of every setting, which is why
it is the default. These runs use $1000$ replications, or $300$ for the close race at $20{,}000$
items, and every error rate was at most $3.2\%$.

\section{Related work in detail}
\label{app:related}

Table~\ref{tab:sota} places the closest methods side by side; the rest of this appendix says in
words what the columns compress, and then relates the construction to the statistical tools on
which it builds.

\begin{table}[t]
\caption{Closest existing methods. Monitoring: valid under continuous looks and any stopping
rule. Finite-sample: no asymptotics. Structure: the correction uses the relations among pairs.}
\label{tab:sota}
\centering\footnotesize\setlength{\tabcolsep}{3pt}
\def\tstack#1{\begin{tabular}[t]{@{}c@{}}#1\end{tabular}}%
\def\lstack#1{\begin{tabular}[t]{@{}l@{}}#1\end{tabular}}%
\begin{tabular}{@{}lccccc@{}}
\toprule
 & \shortstack{Rank sets,\\simultaneous} & Monitoring & Finite-sample & Structure & \shortstack{Early stop\\of evaluation}\\
\midrule
\multicolumn{6}{@{}l}{\emph{Fixed-sample rank inference}}\\
\citet{mogstad2024}                             & yes & no  & no  & no  & no\\
\citet{almohamad2022}                           & yes & no  & yes & no  & no\\
\citet{almohamad2021}                           & yes & no  & yes & yes & no\\
\citet{neuhof2024confident}                     & yes & no  & yes$^{a}$ & no & no\\
\citet{neuhof2026mmlu}                          & yes & no  & yes$^{a}$ & no & no\\
\citet{neuhof2026rank}                          & no  & no  & yes$^{a}$ & no & no\\
\lstack{\citet{li2026lowrank};\\\citet{avelar2026prompt}} & yes & no  & no  & no  & no\\
\citet{chandra2025finite}                       & yes & no  & \tstack{yes (simulable\\noise)} & n/a & no\\
\midrule
\multicolumn{6}{@{}l}{\emph{Sequential evaluation}}\\
\lstack{\citet{zhou2026celeus};\\\citet{hsushekhar2026}}  & no (one model) & yes & yes & n/a & yes\\
\citet{kotawala2026resolution}                  & no (pairs)     & yes & yes & \tstack{no\\correction$^{b}$} & yes (pairwise)\\
\citet{gu2026serpant}, battles                  & yes            & yes & yes & \tstack{no\\(Bonferroni)} & yes (pairwise)\\
\citet{arnold2026smcs}                          & \tstack{no (superior\\set)} & yes & yes & \tstack{no (generic\\closure)} & no\\
\citet{arviv2026stop}                           & no (pairs)     & pre-set looks & no & \tstack{no\\correction} & yes\\
\citet{heckel2019}                              & \tstack{no (partition\\at stopping)} & \tstack{own rule\\only} & yes & \tstack{no (union\\bound)} & yes\\
\midrule
This paper                                      & yes & yes & yes & yes & yes\\
\bottomrule
\end{tabular}

\smallskip
\raggedright\footnotesize
$^{a}$Given finite-sample pairwise tests. $^{b}$Bonferroni and Holm only in its separate
fixed-sample analysis.
\end{table}

Fixed-sample rank sets are valid for one look at one pre-chosen sample size, and the constructions
of \citet{mogstad2024}, \citet{almohamad2022}, \citet{neuhof2024confident} and their leaderboard descendants
\citep{neuhof2026mmlu,neuhof2026rank,li2026lowrank,avelar2026prompt} set a critical value for
a fixed sample size, mostly by asymptotic approximation. They then apply the same counting formula as our
\eqref{eq:LU}. The first difference is the certified set $\D$ to which the formula is applied. Ours
is valid at all times and uses the order structure inside the multiplicity correction. The second
is that enumeration adds the exact rank sets of Theorem~\ref{thm:ranks}(a), which can be strictly
tighter than counting (Theorem~\ref{thm:ranks}(c)). \citet{almohamad2021} instead test one
configuration of the ranks at a time with the partitioning principle, which uses the order
structure and gives rank sets without counting, at one fixed sample size. Our exact test is the
anytime-valid counterpart of their procedure.

Group-sequential stopping needs its looks fixed in advance and a normal approximation, and as
applied so far it has no correction across pairs. Ours needs neither and corrects across all
pairs. \citet{arviv2026stop} report large savings from Pocock boundaries \citep{pocock1977}, one constant critical
value at each pre-scheduled look, on pairwise comparisons. They leave multiplicity across pairs
open, and Appendix~\ref{app:ecost} measures the cost of this. Active-ranking algorithms
\citep{heckel2019} choose which comparisons to make adaptively, which we do not. They guarantee
their answer only at their own stopping time, when they output a grouping of the
models. Ours guarantees every intermediate report under any
stopping rule.

Among the anytime-valid methods, the two closest differ from ours in what they are able to
certify and in how they treat the order structure. \citet{gu2026serpant} control the family-wise
error rate by Bonferroni over the ordered pairs. They use transitivity only to propagate certified
preferences and need transitive preferences only for that step. They assume i.i.d.\ prompts and
give no bound for a fixed benchmark sampled without replacement. In ours transitivity instead follows from scalar
abilities. It restricts the sets of true hypotheses that the multiplicity procedure has to test
at all, which is why the exact test never certifies less than the full e-value closure of
\citet{hartoglei2025}. On the same wealths, the shortcut never certifies less than the
Bonferroni correction of \citet{gu2026serpant}, because its statistic is the pair's own wealth
plus nonnegative terms. \citet{arnold2026smcs}
certify which models may be best, a statement about rank one only, through a closure over all
subsets of models that uses no relation among the pairwise hypotheses. Ours gives simultaneous
rank inference for common-item benchmark scores, with validity on a fixed benchmark that rests
only on the random item order.

The construction builds on several established statistical literatures. Multiple comparisons with the best
\citep{hsu1996} is the classical form of ``which models are tied for first'', and our certified
tier $1$ is its anytime-valid analog. Step-down resampling \citep{romano2005} is the fixed-sample
engine behind \citet{mogstad2024}. Betting martingales \citep{waudbysmith2024}, time-uniform
confidence sequences \citep{howard2021} and Ville's inequality \citep{ville1939} supply the
sequential machinery. \citet{fischerramdas2024} show that online closed testing procedures that
cannot be uniformly improved must use e-values. This is a natural route to models arriving
over time.
On the evaluation side,
\citet{miller2024errorbars} argues for error bars on benchmark scores, computed from item-level
results, \citet{bowyer2025clt} argue
against the normal approximation at small sample sizes, and \citet{chiang2024arena} describe the battle
format we leave out. \citet{singh2025illusion} document how data-dependent practices, such as
testing many private variants and disclosing only the best, shape public leaderboards.
Time-uniform guarantees cover a data-dependent choice of when to stop or report one evaluation,
but not the selection of which model variants to disclose.

\end{document}